\documentclass[11pt]{amsart}

\usepackage{hyperref}

\usepackage{graphicx, enumerate, url}
\usepackage{amssymb,mathtools}
\usepackage{algorithm}
\usepackage{algpseudocode}
\usepackage{color}
\usepackage{tikz}
\usetikzlibrary{3d,calc}
\usepackage{verbatim}
\usepackage{multirow}

\numberwithin{equation}{section}
\numberwithin{algorithm}{section}

\theoremstyle{plain}
\newtheorem{theorem}{Theorem}[section]
\newtheorem{proposition}[theorem]{Proposition}
\newtheorem{lemma}[theorem]{Lemma}
\newtheorem{corollary}[theorem]{Corollary}
\newtheorem{conjecture}[theorem]{Conjecture}

\theoremstyle{definition}
\newtheorem{definition}[theorem]{Definition}
\newtheorem{example}[theorem]{Example}

\theoremstyle{remark}

\let\S\undefined

\DeclareMathOperator{\E}{E}
\DeclareMathOperator{\tr}{Tr}
\DeclareMathOperator{\Tr}{Tr}
\DeclareMathOperator{\diag}{diag}
\DeclareMathOperator{\rank}{rank}

\DeclareMathOperator{\V}{V}

\DeclareMathOperator{\S}{S}

\newcommand{\N}{\mathbb{N}}
\newcommand{\C}{\mathbb{C}}

\newcommand{\R}{\mathbb{R}}

\newcommand{\g}{\mathbf{g}}
\newcommand{\x}{\mathbf{x}}

\newcommand{\y}{\mathbf{y}}
\newcommand{\z}{\mathbf{z}}

\renewcommand{\V}{\mathbf{V}}
\newcommand{\W}{\mathbf{W}}
\newcommand{\cA}{\mathcal{A}}

\newcommand{\cC}{\mathcal{C}}
\newcommand{\cD}{\mathcal{D}}

\newcommand{\cH}{\mathcal{H}}

\newcommand{\cS}{\mathcal{S}}
\newcommand{\cT}{\mathcal{T}}
\newcommand{\cU}{\mathcal{U}}

\newcommand{\ba}{\mathbf{a}}
\newcommand{\bA}{\mathbf{A}}
\newcommand{\bb}{\mathbf{b}}
\newcommand{\bB}{\mathbf{B}}
\newcommand{\bc}{\mathbf{c}}
\newcommand{\bC}{\mathbf{C}}

\newcommand{\be}{\mathbf{e}}
\newcommand{\bE}{\mathbf{E}}
\newcommand{\bF}{\mathbf{F}}
\newcommand{\bbf}{\mathbf{f}}
\newcommand{\bg}{\mathbf{g}}
\newcommand{\bG}{\mathbf{G}}

\newcommand{\bH}{\mathbf{H}}
\newcommand{\bi}{\mathbf{i}}
\newcommand{\bI}{\mathbf{I}}
\newcommand{\bM}{\mathbf{M}}
\newcommand{\bp}{\mathbf{p}}

\newcommand{\bu}{\mathbf{u}}
\newcommand{\bv}{\mathbf{v}}
\newcommand{\bw}{\mathbf{w}}

\newcommand{\by}{\mathbf{y}}
\newcommand{\bz}{\mathbf{z}}

\newcommand{\rA}{\mathrm{A}}
\newcommand{\rB}{\mathrm{B}}

\newcommand{\rS}{\mathrm{S}}

\newcommand{\rU}{\mathrm{U}}

\newcommand{\0}{\mathbf{0}}
\newcommand{\1}{\mathbf{1}}

\newcommand{\vc}{\vcentcolon =}

\newcounter{mnotecount}[section]
\renewcommand{\themnotecount}{\thesection.\arabic{mnotecount}}
\newcommand{\mnote}[1]
{\protect{\stepcounter{mnotecount}}$^{\mbox{\footnotesize
$
\bullet$\themnotecount}}$ \marginpar{
\raggedright\tiny\em
$\!\!\!\!\!\!\,\bullet$\themnotecount: #1} }

\definecolor{dg}{rgb}{0,.5,0} 

\newcommand{\raf}[1]{{\color{blue}RB: #1}}
\newcommand{\cv}{= \vcentcolon}

\begin{document}
\title{A new class of distances on complex projective spaces}
\author[R. Bistro\'n]{Rafa{\l} Bistro\'n}
\address{Institute of Theoretical Physics, Faculty of Physics, Astronomy and Applied Computer Science, Jagiellonian University, Krakow, Poland}
\address{Doctoral School of Exact and Natural Sciences, Jagiellonian University}
\email{rafal.bistron@doctoral.uj.edu.pl}
\author[M. Eckstein]{Micha\l $\,$ Eckstein}
\address{Institute of Theoretical Physics, Faculty of Physics, Astronomy and Applied Computer Science, Jagiellonian University, Krakow, Poland}
\email{michal.eckstein@uj.edu.pl}
\author[S. Friedland]{Shmuel~Friedland}
\address{Department of Mathematics, Statistics, and Computer Science, University of Illinois at Chicago,  Chicago, Illinois, 60607-7045, USA}
\email{friedlan@uic.edu}
\author[T. Miller]{Tomasz Miller}
\address{Copernicus Center for Interdisciplinary Studies, Jagiellonian University, Krakow, Poland}
\email{tomasz.miller@uj.edu.pl}
\author[K.~{\.Z}yczkowski]{Karol~{\.Z}yczkowski}
\address{Institute of Theoretical  Physics, Faculty of Physics, Astronomy and Applied Computer Science, Jagiellonian University, Krakow, Poland}
\address{Center for Theoretical Physics, Polish Academy of Science, Warsaw, Poland}
\email{karol.zyczkowski@uj.edu.pl}
\date{December 5,  2023}
\begin{abstract}
The complex projective space $\mathbb{P}(\C^n)$
can be interpreted as the space of all quantum pure states
of size $n$. A distance on this space,
interesting from the perspective of quantum physics, 
can be induced from a classical distance defined on the $n$-point probability simplex by the `earth mover problem'.
We show that this construction
leads to a quantity satisfying the triangle inequality,
which yields a true distance on complex projective space
belonging
to the family of quantum $2$-Wasserstein distances.
\end{abstract}
\maketitle

\maketitle

\section{Introduction}
Assume that $\cH$ is an $n$-dimensional space over the complex numbers $\C$.
A pure state $|\phi\rangle \in \cH$ is a vector of length one, which we identify with 
$e^{\bi\theta}|\phi\rangle$, with $\bi^2 \vc-1$ and $\theta\in\R$.  After introducing an orthonormal basis $\be_i=| i\rangle$, for $i\in[n]\vc\{1,\ldots,n\}$, we identify $\cH$ with $\C^n$ equipped with the standard inner product, $\langle \phi|\psi\rangle=\langle \x,\y\rangle=\x^*\y$. Then, the space of pure states in $\cH \equiv \C^n$ is identified with the complex projective space $\mathbb{P}(\C^n)$ of complex  dimension $n-1$. 



$\mathbb{P}(\C^n)$ can be equipped with a variety of different distance functions, for instance the one induced by the Fubini--Study metric. Another standard choice, known in physics literature as the Hilbert--Schmidt distance, arises as follows:
%
Let $\rB(\cH)$ be the space of linear operators on $\cH$.    
For $T\in\rB(\cH)$ denote by $\|T\|_F=\sqrt{\tr T T^*}$ the Frobenius (aka Hilbert--Schmidt) norm.  On $\rB(\cH)$ define a distance
$d_{\text{HS}}(S,T)\vc\frac{1}{\sqrt{2}}\|S-T\|_F$.  
For states in $\cH$, which can be identified with one-dimensional projectors in $\rB(\cH)$ as follows: $\mathbb{P}(\C^n) \ni \x  \sim \x \x^* \in \rB(\C^n)$, the above distance reduces to
\begin{equation}\label{basdistPC}
d_{\text{HS}}(\x,\y)=\sqrt{1-|\x^*\y|^2}=\sqrt{(\x^*\x)(\y^*\y)-|\x^*\y|^2}. 
\end{equation}

\medskip

The goal of this paper is to consider more general types of distances, such that $d(\x,\y)^2$ is a quadratic form on the wedge product $\x\wedge\y$ in the exterior product space $\Lambda^2(\cH)$, which can be identified with the space of skew-symmetric matrices in $\C^{n\times n}$.  
Namely,  assume that $\bE$ is a self-adjoint positive semidefinite operator on $\Lambda^2(\cH)$.  Define
\begin{equation}\label{distL2H}
d_\bE(\x,\y) \vc \|\bE(\x\wedge\y)\|,\quad \x,\y\in\cH.
\end{equation}
We study the following problem: For which $\bE$ does formula \eqref{distL2H} yield a distance on $\mathbb{P}(\C^n)$?

The study of distances on complex projective spaces is of paramount interest for quantum information processing \cite{BZ17,Keyl,NielsenChuang}. They serve as a measure of proximity of quantum states and have multifarious applications in quantum metrology and sensing \cite{QSense,QMetro}, and quantum machine learning \cite{QML,KdPMLL22}. 

The distances generated by the
Monge transport problem  
enjoy an important classical property:
the distance between any two spin-coherent states,
each concentrated around a given point on the sphere,
is equal to the geodesic distance between
these points \cite{ZS01}.

A generalised approach to the transport problem by
Kantorovich \cite{Kan48}
allowed one to introduce several classes of 
distances based on the {\sl quantum transport problem}
\cite{CGP20,GP18,ZYYY22}.
A more general approach by Wasserstein
was also recently applied
\cite{Duv22,PMTL21,PT21}
to  construct distances in the space of quantum states.

In this work analyze a  particular class of distances \eqref{distL2H} directly connected to the quantum optimal transport problem --- see \cite{CEFZ21,BEZ22,FECZ21}.
Interestingly, the problem of proving the triangle 
inequality in its general form stated in Conjecture \ref{conj}
is related to the problem of describing the
set in the complex plane 
covered by diagonal product of special unitary matrices
\cite{Mi23}.

\section{\label{sec:survey}Survey of the results}


The function $d_\bE$ given by formula \eqref{distL2H} is manifestly symmetric and non-negative. Furthermore, it is non-degenerate, and hence a \emph{semi-distance} if and only if $( \ker \bE ) \cap ( \cH \wedge \cH ) = \{0\}$.
%
Therefore, our central question boils down to the problem: For which $\bE$ does formula \eqref{distL2H} satisfy the \textit{triangle inequality},
\begin{equation}\label{triangins}
\|\bE(\x\wedge\y)\| \leq \|\bE(\x\wedge\z)\| + \|\bE(\z\wedge\y)\|, \quad \text{for all } \x,\y,\z \in \mathbb{P}(\C^n)?
\end{equation}
In Section \ref{sec:triangin} we show (Corollary \ref{cornecsufcondti}) that it is sufficient to prove \eqref{triangins} for every triple of orthonormal vectors $\x,\y,\z$ in $\cH$.

We shall call the operators satisfying \eqref{triangins} the \emph{triangular operators}. In Section \ref{sec:triangin} we provide a necessary and sufficient condition for an operator $\bE$ to be triangular (Theorem \ref{necsufcondti}) in terms of the eigenvalues of 3-dimensional projectors. We also present a simple sufficient condition (Proposition \ref{sufcond}) based on the eigenvalues of $\bE$. Then, in Section \ref{sec:to} we show (Proposition~\ref{cT2cone}) that the set of all triangular operators on a given Hilbert space $\cH$ has the structure of a 2-cone. We study its symmetries and extreme rays (Theorem~\ref{ern=3}).

A special class of semi-distances arises when the operator $\bE$ is of the form
\begin{equation}\label{Euiujeig}
\bE(\bu_i\wedge\bu_j)=d_{ij} \, \bu_i\wedge\bu_j,
\end{equation}
for some $d_{ij}\ge 0$ and $\langle \bu_i,\bu_j\rangle=\delta_{ij}$ for all $i,j\in[n]$. The induced semi-distance $d_\bE$ is called the \emph{quantum 2-Wasserstein semi-distance} \cite{CEFZ21,FECZ21}. These are discussed in Section~\ref{sec:staoI}.  
Let $D(\bE)=[d_{ij}]\in\R^{n\times n}_+$ be the symmetric non-negative matrix with zero diagonal, induced by the eigenvalues of $\bE$. It is easy to see that $\bE$ of the form \eqref{Euiujeig} is a triangular operator only if $D(\bE)$ is a \emph{distance matrix} (see Definition \ref{defdsitmat}), i.e.
\begin{align}\label{tri_lambda}
    d_{ij} \leq d_{ik} + d_{kj}, \quad \text{ for all } i,j,k \in [n].
\end{align}
Indeed, if \eqref{tri_lambda} is violated for a triple of indices $i,j,k$ then \eqref{triangins} fails for $\x=\bu_i$, $\y=\bu_j$, $\z=\bu_k$. 

Basing on extensive analytical and numerical studies we conjecture that \eqref{tri_lambda} is actually a necessary \emph{and sufficient} condition for an operator $\bE$ of the form \eqref{Euiujeig} to be triangular.

\begin{conjecture}\label{conj}
An operator $\bE$ of the form \eqref{Euiujeig} is triangular if and only if $D(\bE)$ is a distance matrix.
\end{conjecture}

In other words, Conjecture \ref{conj} claims that every operator triangular $\bE$ of the form \eqref{Euiujeig} is induced by some distance matrix (cf. Definition \ref{def:ind}). 

For $\cH = \C^2$ there is only one operator $\bE$, up to a multiplicative prefactor, it is of the form \eqref{Euiujeig} and it is trivially a triangular operator. For $\cH = \C^3$ every operator $\bE$ is of the form \eqref{Euiujeig} (see Corollary \ref{cor:E3dim}). In Section \ref{sec:triangin} we prove (Corollary \ref{cornecsufcondti3D}) that for $\dim \cH = 3$ the conjecture holds.

In Section \ref{sec:staoI} we show the validity of Conjecture \ref{conj} for a wide class of operators. Specifically, Theorem \ref{subscH} proves Conjecture \ref{conj} for all operators $\bE$ induced by the Euclidean distance, i.e. operators, the eigenvalues of which are given by $d_{ij} = \| \bp_i - \bp_j \|_2$ for some $n$-tuple of distinct points $\bp_1, \ldots, \bp_{n} \in \R^n$, for any $n = \dim \cH$. Furthermore, Theorem \ref{thm:reducible} shows that  Conjecture \ref{conj} holds if $d_{ij} = \lambda_i \lambda_j$, that is if $\bE = O \wedge O$ for some operator $O \in \rB(\cH)$. Finally, in Section \ref{subscH3} we present an example, based on Proposition \ref{sufcond}, of a triangular operator induced by the $\ell_1$ distance.

Lastly, in Appendix \ref{sec:numerics} we provide numerical evidence for the general validity of Conjecture \ref{conj}. Appendix \ref{sec:lemma} contains a technical lemma needed in the proof of Theorem \ref{minxyzthm}.

\section{Preliminary results} \label{sec:prelimres}
\subsection{Certain properties of eigenvalues of self-adjoint operators}\label{subsec:propselad}
Denote $\rB(\cH)\supset\rS(\cH)\supset \rS_+(\cH)$ the space of linear operators of $\cH$ to itself, the self-adjoint operators $T^*=T$, and the cone of positive semidefinite  operators, respectively.  For $T\in\rB(\cH)$ let $\|T\|_F=\sqrt{\tr T T^*}$ be the Frobenius norm.  Assume that $T\in\rS(\cH)$.
Assume that $\bu_1,\ldots,\bu_n$ is an orthonormal  set of eigenvectors of $T$ with the corresponding eigenvalues:
\begin{equation*}
\lambda_{\max}(T)=\lambda_1(T)\ge \ldots\ge\lambda_n(T)=\lambda_{\min}(T).
\end{equation*}
When there is no ambiguity concerning $T$ we shall simply write $\lambda_i = \lambda_i(T)$.
%
\subsection{2-tensor product and 2-exterior algebra}\label{subsec:2ext}

Let $\cH^{\otimes 2}=\cH\otimes \cH$.  We endow $\cH^{\otimes 2}$ with the induced inner product 
\begin{equation*}
\langle \x\otimes \bu,\y\otimes \bv\rangle=\langle \x,\y\rangle \langle \bu,\bv\rangle.
\end{equation*}
We identify $\cH\otimes\cH$ with $\C^{n\times n}$ via the correspondence $ \x\otimes\y\sim \x\y^\top$. Then the inner product on $\C^{n\times n}$ is given by $\langle A,B\rangle=\tr A^* B$.  The space $\C^{n\times n}$ decomposes as a direct sum of symmetric $\rS^2\C^n$ and skew-symmetric matrices $\rA^2\C^n$ respectively.  
Note that these two subspaces are orthogonal,  and invariant under the conjugation $A\mapsto A^*$.  The space $\cH^{\otimes 2}$ is a direct sum of symmetric and skew-symmetric tensors $\rS^2(\cH)\oplus \Lambda^2(\cH)$.  In mathematics, $\Lambda^2(\cH)$ is called the $2$-exterior product of $\cH$, while in physics it is the $2$-fermion space.  An ``irreducible'' element in $\rS^2(\cH)$ is a rank-one tensor $\x\otimes \x\ne \0$, which corresponds to a symmetric rank-one matrix.
An ``irreducible'' element in $\Lambda^2(\cH)$ corresponds to 
\begin{equation}\label{defxwedy}
\x\wedge\y=\frac{1}{\sqrt{2}}(\x\otimes \y-\y\otimes \x)\ne \0,
\end{equation}
which corresponds to a rank-two skew-symmetric matrix.
Using the Lagrange identity
\begin{equation}\label{Lagrangeid}
\|\x\wedge\y\|^2=\|\x\|^2\|\y\|^2 -|\langle \x,\y\rangle|^2
\end{equation}
one can write the Hilbert--Schmidt distance \eqref{basdistPC} on $\mathbb{P}(\C^n)$ as $d_{\text{HS}}(\x,\y) = \|\x\wedge\y\|$.
Note that  $\x\wedge\y\ne \0$ generates a two-dimensional vector space in $\cH$.  If $\dim\cH=3$,  then we can identify $\x\wedge\y$ with the vector orthogonal to span$(\x,\y)$, whose length is $\|\x\wedge\y\|$.  
In this case, the vector $\x\wedge\y$ can be identified with the cross product $\x\times\y$.   (One can define the cross product $\times$ on $\C^3$ as in the real case.)

Assume that $\W\subset \cH$ is a $k$-dimensional subspace.  Denote by $\Lambda^2(\W)\subset \Lambda^2(\cH)$ the subspace generated by $\x\wedge\y$ for $\x,\y\in\W$.  Clearly, $\dim \Lambda^2(W)={\dim\W\choose 2}$.  Let  $\bA\in \rS(\Lambda^2(\cH))$.   Assume that $\W$ is a nontrivial subspace of $\cH$.   Then by restriction of $\bA$ to $\Lambda^2(\W)$ we denote an operator $\bB\in\rS(\Lambda^2(\W))$ such that
\begin{equation*}
\langle \bB X,Y\rangle=\langle \bA X,Y\rangle \textrm{ for all } X,Y\in \Lambda^2(\W).
\end{equation*}
Note that if $\bA$ is positive (semidefinite) then $\bB$ is positive (semidefinite).

Assume that $\dim\cH=n\ge 2$.  Let $\bu_1,\ldots,\bu_n$ be an orthonormal basis in $\cH$.  Then $\bu_i\wedge\bu_j, 1\le i<j\le n$ is an orthonormal basis in $\cH\wedge\cH$.  Note that if $\bu_i\wedge\bu_j$ is viewed as a skew-symmetric matrix, then its rank equals~$2$.  Let $A_1,\ldots,A_{n(n-1)/2}$ be an orthonormal basis in $\Lambda^2(\cH)$.  Then, for $n>3$, the vectors $A_1,\ldots,A_{n(n-1)/2}$ need not to be irreducible.  Indeed, we can choose $A_1$ corresponding to a skew-symmetric matrix of maximum rank $2\lfloor n/2\rfloor$.  The case $n=3$ is special:
\begin{lemma}\label{wdge3dim}  Let $\cH$ be a three dimensional subspace of a Hilbert space over $\C$.  Then:
\begin{enumerate}[(a)]
\item Any vector in $\Lambda^2(\cH)$ is of the form $\x\wedge \y$.
\item Any orthonormal basis in $\Lambda^2(\cH)$ is of the form $\bu_1\wedge\bu_2,
\bu_2\wedge\bu_3,\bu_3\wedge\bu_1$, where $\bu_1,\bu_2,\bu_3$ is an orthonormal basis in $\cH$.
\end{enumerate}
\end{lemma}
\begin{proof}  It is enough to assume that $\cH=\C^3$. 
Then $\Lambda^2(\C^3)$ is viewed as the space of skew-symmetric matrices 
\begin{equation*}
\rA^2\C^3=\left \{X\in \C^{3\times 3},
X=\begin{bmatrix}0&x_1&x_2\\-x_1&0&x_3\\-x_2&-x_3&0\end{bmatrix} \right\}.
\end{equation*}
The inner product on $\rA^2\C^3$ is $\langle X, Y\rangle=\tr X^*Y$.
Assume that $X\ne 0$.  Recall that since $X$ is skew-symmetric it follows that $\det X=0$.  That is, at least one eigenvalue of $X$ is zero.   Let $\bu_3, \|\bu_3\|=1$ be the corresponding eigenvector of $X$: $X\bu_3=0$.   Complete $\bu_3$ to an orthonormal basis $\bu_1,\bu_2, \bu_3$ in $\C^3$.  Let $U=[\bu_1\bu_2\bu_2]\in\C^{3\times 3}$.  Note that $U$ is unitary.  Set $Y=U^\top X U=[U^\top [(X\bu_1) (X\bu_2) \0]$.  Note that $Y$ is skew-symmetric and the last column of $Y$ is zero.  Hence,  
\begin{align*}
Y & =\begin{bmatrix} 0&y_1&0\\-y_1&0&0\\0&0&0\end{bmatrix}=y_1(\be_1\be_2^\top-\be_2\be_1^\top)=y_1 \sqrt{2}\be_1\wedge \be_2,\\
X & =y_1\sqrt{2}(\bar U\be_1)\wedge (\bar U\be_2).
\end{align*}
This shows part (a).

Observe that $\langle X, X\rangle =1$ iff $|y_1|=1/\sqrt{2}$.  By considering $U^\top X U$ we can assume as above that
\begin{equation*}
X=\be_1\wedge\be_2=\frac{1}{\sqrt{2}}\begin{bmatrix} \,\,0&1&0\\-1&0&0\\\,\,0&0&0\end{bmatrix}.
\end{equation*}
Suppose we make a change of coordinates in the subspace spanned by $\be_1,\be_2$:
\begin{equation*}
\bu_i=\sum_{j=1}^2 u_{ij}\be_j, \quad i\in[2],  \quad U=[u_{ij}]\in\C^{2\times 2}.
\end{equation*}
Then $\bu_1,\bu_2$ are two orthonormal vectors if and only if $U$ is unitary.
It is straightforward to show that $\bu_1\wedge \bu_2=\det U \be_1\wedge \be_2$.

Assume now that $Z\in \rA^2\C^3$ has norm one and is orthogonal to $X$. 
So $Z=\z_1\wedge \z_2$, where $\z_1,\z_2$ are orthogonal and $\|\z_1\|=\|\z_2\|=1$.
Recall that the intersection of the subspaces $\W_1= \text{span}(\be_1,\be_2)$ and 
$\W_2= \text{span}(\z_1,\z_2)$  are either equal 
or  their intersection contains exactly one line spanned by the vector $\y, \|\y\|=1$.  Let us now change the bases in $\W_1$ and $\W_2$ into $\{\y,\bv\}$ and $\{\y,\bw\}$, respectively.  If $\W_1=\W_2$ then we choose $\bv=\bw$ and $X$ and $Z$ are not orthogonal.  Hence $\W_1\ne \W_2$.  Again, without loss of generality we may assume that $\y=\be_2, \bv=\be_1$.  Hence $\bw=a\be_1+b\be_3$.  The condition $\tr X^*Z=0$ yields that $a=0$.  Hence we can assume that $\bw=\be_3$.  Let us now pick a vector $V\in \rA^2\C^3$ of length one orthogonal to both $X$ and $Z$.  Representing $X$ and $Z$ as skew-symmetric matrices we deduce that 
 that $V=\zeta\be_1\wedge \be_3$, where $\zeta\in\C$ and $|\zeta|=1$. 
\end{proof}

The following result is an immediate consequence of Lemma \ref{wdge3dim}.
\begin{corollary}\label{cor:E3dim}
If $\cH = \C^3$ then every operator $\bE\in\rB \big(\Lambda^2(\cH) \big)$ is of the form \eqref{Euiujeig}.
\end{corollary}
%
\subsection{The action of $\Lambda^2(\cH)$ on $\cH$ and a variational formula}\label{subsec: actL2H}
Let $C\in\Lambda^2(\cH)$.  We associate with $C$ an anti-linear operator $C^\vee$ on  $\rB(\cH)$:
\begin{equation*}
C^\vee(a\x+b\y)=\bar a C^\vee(\x)+\bar b C^\vee(\y),
\end{equation*} 
such that the following equality holds
\begin{equation}\label{tildeAdef}
\langle C^\vee(\y),\x\rangle=\langle C, \x\wedge \y\rangle \textrm{ for } \x,\y\in \cH.
\end{equation}
We shall only be interested in the operators of the form $C=\ba\wedge\bb$, for some $\ba,\bb \in \cH$.  Recall that
\begin{equation*}
\langle \ba\wedge\bb,\x\wedge\y\rangle=\langle \ba,\x\rangle\langle \bb,\y\rangle-
\langle \ba,\y\rangle\langle \bb,\x\rangle
\end{equation*}
and hence
\begin{equation}\label{tildeAdef1}
(\ba\wedge\bb)^\vee(\y)=\langle \y,\bb\rangle \ba -\langle\y,\ba\rangle \bb.
\end{equation}
Note also that
\begin{equation*}
\langle C^\vee(\x),\y\rangle=-\langle C,\x\wedge \y\rangle.
\end{equation*}

In what follows we need to consider the following variation of a quadratic form in $\x\wedge\y$:
\begin{lemma}\label{varformM} Let $\bM \in \rS(\Lambda^2(\cH))$. Then
\begin{multline}\label{varformM1} 
d\left\langle \bM(\x\wedge\y) ,\x\wedge \y \right\rangle= \\ =2\Re \left(\left\langle \big(\bM(\x\wedge\y) \big)^\vee(\y),d\x\right\rangle-\left\langle \big(\bM(\x\wedge\y) \big)^\vee(\y),d\y\right\rangle\right).
\end{multline}
Suppose furthermore that $\bM\in\rS_+(\Lambda^2(\cH))$, $\x\wedge\y\ne \0$ and $\|\sqrt{\bM}(\x\wedge\y)\|>0$.  Then
\begin{equation}\label{varformM2} 
d\|\sqrt{\bM}(\x\wedge\y)\|=\\
\frac{1}{\|\sqrt{\bM}(\x\wedge\y)\|} \; d\left\langle \bM(\x\wedge\y) ,\x\wedge \y \right\rangle.
\end{equation}
\end{lemma}
The proof of the lemma is straightforward and we leave it to the reader.
\begin{proposition}\label{propMxy} Let $\bM \in\rS(\Lambda^2(\cH))$ and consider the critical points of  $\langle \bM(\x\wedge\y),\x\wedge \y\rangle$ on the product state space $\mathbb{P}(\C^n)\times \mathbb{P}(\C^n)$.  Then $(\x^\star,\y^\star)\in \mathbb{P}(\C^n)\times \mathbb{P}(\C^n)$ is a critical point if and only if the following
equalities hold
\begin{equation}\label{propxMxy1} 
\begin{aligned}
& \big(\bM(\x^\star\wedge\y^\star)\big)^\vee(\y^\star)= \langle \bM(\x^\star\wedge\y^\star),\x^\star\wedge \y^\star\rangle \x^\star, && \|\x^\star\|=1,\\
& \big(\bM(\x^\star\wedge\y^\star)\big)^\vee(\x^\star)= -\langle \bM(\x^\star\wedge\y^\star),\x^\star\wedge \y^\star\rangle \y^\star, && \|\y^\star\|=1.
\end{aligned}
\end{equation}
In particular, for $\lambda=\langle \bM(\x^\star\wedge\y^\star),\x^\star\wedge \y^\star\rangle$ and any $\bu,\bv\in\C^n$ we have
\begin{equation}\label{propxMxy2}  
\begin{aligned}
& \langle\bM(\x^\star\wedge\y^\star)-\lambda\x^\star\wedge\y^\star,\bu\wedge\y^\star\rangle=\\
& = \langle\bM(\x^\star\wedge\y^\star)-\lambda\x^\star\wedge\y^\star,\x^\star\wedge\bv\rangle=0.
\end{aligned}
\end{equation}
\end{proposition}
\begin{proof} Since we impose the conditions $\|\x\|=\|\y\|=1$ we introduce the Lagrange multipliers $\lambda\langle \x,\x\rangle, \mu\langle \y,\y\rangle$, with $\lambda,\mu\in\R$.  Observe that 
\begin{equation*}
\lambda d\langle \x,\x\rangle=2\lambda\Re\langle \x,d\x\rangle,\quad \mu d\langle \y,\y\rangle=2\mu\Re\langle \y,d\y\rangle.
\end{equation*}
Combine the above identity with \eqref{varformM1} to deduce \eqref{propxMxy1}.
The equalities \eqref{varformM2}  follows straightforwardly from \eqref{varformM1} and 
\eqref{tildeAdef}.
\end{proof}
Note the the variety $\{(\x,\y)\in\mathbb{P}(\C^n)^2, \x\wedge\y=\0\}$ is a (trivial) variety of critical points of $\langle \bM(\x\wedge\y),\x\wedge \y\rangle$.
%
%
%
\subsection{Wedge product of operators on $\cH$}\label{subsec:wpop}
Assume that $T\in\rB(\cH)$. Then $\rS^2(\cH)$ and $\Lambda^2(\cH)$ are invariant subspaces of  $T\otimes T\in \rB(\cH\otimes \cH)$,
\begin{equation*}
\begin{aligned}
& (T\otimes T)(\x\otimes\y)=(T\x)\otimes (T\y),\\
& (T\wedge T)(\x\wedge\y)=(T\x)\wedge(T\y).
\end{aligned}
\end{equation*}
Assume that $T$ is represented as a matrix $A\in\C^{n\times n}$ in the standard orthonormal basis in $\C^n$.  Then $T\wedge T$ is represented \cite{Frib} by the 2-compound matrix $C_2(A)\in \C^{{n \choose 2}\times {n\choose 2}}$, where
\begin{align*}
\big( C_2(A) \big)_{(i,j),(p,q)} = \det\begin{bmatrix} a_{ip}&a_{iq}\\a_{jp}&a_{jq}\end{bmatrix}, 
\end{align*}
for $1\le i<j\le n, 1\le p<q\le n$. It is convenient to denote $C_2(A)$ by $A\wedge A$.

Let $\bbf_1,\ldots,\bbf_n$ be another orthonormal basis in $\C^n$.  Then $\bbf_i=U\be_i, i\in[n]$ for some unitary $U\in \rU(\C^n)$.  Then $T\wedge T$
is represented in the new basis by $(U^*\wedge U^*)(A\wedge A)(U\wedge U)$ and $\bu_i\wedge\bu_j$, for $1\le i<j\le n$, is an orthonormal set of eigenvectors of $T\wedge T$ with the corresponding eigenvalues $\lambda_i(T)\lambda_j(T)$.  If
$T\in\rS_+(\cH)$ than 
\begin{equation*}
\lambda_{\max}(T\wedge T)=\lambda_1(T)\lambda_2(T), \quad \lambda_{\min}(T\wedge T)=\lambda_{n-1}(T)\lambda_n(T).
\end{equation*}

%
\section{The triangle inequality}\label{sec:triangin}
In this section we study some general necessary and sufficient conditions for an operator $\bE\in \rB(\Lambda^2(\cH))$ to be triangular. Observe that for any $\bE\in\rB(\Lambda^2(\cH))$ one has the equality
\begin{equation*}
\|\bE(\x\wedge\y)\|=\big\|\sqrt{\bE^*\bE}(\x\wedge\y)\big\|.
\end{equation*}
Hence, we can assume without loss of generality that $\bE\in\rS_+(\Lambda^2(\cH))$.

The main result of this section is the following:
\begin{theorem}\label{necsufcondti}  Let $\cH$ be a finite dimensional vector space of dimension at least $3$.  Assume that $\bE\in\rS_+(\Lambda^2(\cH))$.  Then, the triangle inequality holds \eqref{triangins}
if and only if the following condition holds: 
\smallskip \\
\indent For any 3-dimensional subspace $\W\subseteq \cH$, with $\bF(\W)$ denoting the restriction of $\bE^2$ to $\Lambda^2(\W)$, we have:
\begin{equation}\label{necsufcondti1}
\sqrt{d_{12}(\bF(\W))} \leq \sqrt{d_{23}(\bF(\W))}+\sqrt{d_{13}(\bF(\W))}.
\end{equation}
\end{theorem}

In next subsections we provide the intermediate steps of the proof and the final proof is given on page \pageref{proof_thm41}. For the most part we discuss the case when $\bE$ is positive definite.
Then, the case of a positive semidefinite $\bE$ can be viewed as a limiting case of positive definite $\bE$.

An immediate consequence of Theorem \ref{necsufcondti} is the proof of validity of Conjecture \ref{conj} for $\dim \cH = 3$.
\begin{corollary}\label{cornecsufcondti3D}
Let $\dim \cH =3$. An operator $\bE\in\rS_+(\Lambda^2(\cH))$ is triangular if and only if $d_{12}(\bE) \leq d_{23}(\bE) + d_{13}(\bE)$.
\end{corollary}

\subsection{The case where $\x,\y,\z$ are linearly dependent}
In this subsection we show that if $\x,\y,\z$ are in a two-dimensional $\W\subseteq\cH$
then the triangle inequality holds.
\begin{proposition}\label{proplindep} Assume that $\bE\in \rS_+(\Lambda^2(\C^n))$ and $\x,\y,\z\in\mathbb{P}(\C^n)$ are linearly dependent.   Then the triangle inequality \eqref{triangins} holds.  
Equality holds if either $\x\wedge\z=0$ or $\y\wedge\z=0$,  or span$(\x,\y,\z)$ is a 2-dimensional subspace $\W$ and the restriction of $\bE^2$ to $\Lambda^2(\W)$ is zero.
\end{proposition}
\begin{proof} If $\x\wedge\z=\0$ then $\x$ and $\z$ are colinear and equality in \eqref{triangins} holds.  Similarly, if $\y\wedge \z=\0$ equality in \eqref{triangins} holds. 
Assume that  $\x\wedge\z\ne\0$ and $\y\wedge \z\ne\0$.  Without loss of generality we can assume that 
\begin{equation*}
\W=\textrm{span}(\x,\y)=\textrm{span}(\be_1,\be_2 ), \quad \be_1=(1,0)^\top,\be_2=(0,1)^\top.
\end{equation*}
Let $\bA$ be the restriction of $\bE^2$ to $\Lambda^2(\W)$.   If $\bA=0$ then the triangle inequality holds trivially. 
It is left to consider the case where $\bA$ is identity on the one dimensional space $\Lambda^2(\W)$.  That is $\|\bE(\bu\wedge\bv)\|=\|\bu\wedge\bv\|$ for $\bu,\bv\in \W$.
Without loss of generality we can assume that
\begin{equation*}
\begin{aligned}
&\x=(1,0)^\top,  \; \by=(y_1,y_2)^\top,  |y_1|^2+|y_2|^2=1, \; \bz=(z_1,z_2)^\top,  |z_1|^2+|z_2|^2=1,\\
&\|\x\wedge\y\|=|y_2|, \quad \|\x\wedge\z\|=z_2>0, \quad \|\y\wedge\z\|=|y_1 z_2-y_2z_1|>0.
\end{aligned}
\end{equation*}
If $z_2\ge y_2$ one has strict inequalities in \eqref{triangins}.
To show the inequality \eqref{triangins} it is enough to consider the case $1\ge y_2>z_2>0$.  It is straightforward to show that
\begin{multline*}
z_2+|y_1z_2-y_2z_1|-y_2\ge z_2-y_2+ |z_1|y_2-|y_1|z_2=\\
=z_2\left(1-\sqrt{1-y_2^2}\right)-y_2\left(1-\sqrt{1-z_2^2}\right)=\\
=y_2z_2\left(\frac{y_2}{1+\sqrt{1-y_2^2}}-\frac{z_2}{1+\sqrt{1-z_2^2}}\right).
\end{multline*}
It is left to show that 
\begin{equation*}
f(y_2,z_2)\vc\frac{y_2}{1+\sqrt{1-y_2^2}}-\frac{z_2}{1+\sqrt{1-z_2^2}}>0, \quad \textrm{ for }y_2>z_2>0.
\end{equation*}
Fix $z_2$ and observe that $f(y_2,z_2)$ increases on $[0,1]$ in $y_2$.    Hence
$f(y_2,z_2)>f(z_2,z_2)=0$ for $y_2\in (z_2,1)$. 
\end{proof}
\subsection{A minimum problem}\label{subsec:minprob}
Assume that $n = \dim \cH \ge 3$ and let $\bE\in\rS_+(\Lambda^2(\cH))$ be positive definite.
Consider the minimum problem
\begin{equation}\label{minprob}
\begin{aligned}
&\mu(\bE)=\min_{\|\x\|=\|\y\|=\|\z\|=1} f(\x,\y,\z),  \\
& f(\x,\y,\z)=\|\bE(\x\wedge\z)\|+\|\bE(\y\wedge\z)\|-\|\bE(\x\wedge\y)\|.
\end{aligned}
\end{equation}
Clearly, the triangle inequality holds if and only if $\mu(\bE)=0$. 
The following lemma follows straightforwardly from Lemma \ref{varformM}:
\begin{lemma}\label{dfxyzfor}
Assume that $n\ge 3$ and $\bE\in\rS_+(\Lambda^2(\C^n))$ is positive definite.
Let $f(\x,\y,\z)$ be defined as in \eqref{minprob}.   Assume that $\x,\y,\z\in\C^n$ are linearly independent. Then
\begin{equation}\label{dfxyzfor1}
\begin{aligned}
& df(\x,\y,\z)=\\
& \quad \frac{1}{\|\bE(\x\wedge\z)\|}\Re\Big(\big \langle (\bE^2(\x\wedge\z))^\vee(\z),d\x\big\rangle-\big\langle (\bE^2(\x\wedge\z))^\vee(\x),d\z\big\rangle\Big)+\\
& \quad \frac{1}{\|\bE(\y\wedge\z)\|}\Re\Big(\big\langle (\bE^2(\y\wedge\z))^\vee(\z),d\y\big\rangle-\big\langle (\bE^2(\y\wedge\z))^\vee(\y),d\z\big\rangle\Big)+\\
& \quad \frac{1}{\|\bE(\x\wedge\y)\|}\Re\Big(\big\langle (\bE^2(\x\wedge\y))^\vee(\x),d\y\big\rangle-\big\langle (\bE^2(\x\wedge\y))^\vee(\y),d\x\big\rangle\Big).
\end{aligned}
\end{equation}
\end{lemma}
\begin{lemma}\label{critptsf}
Assume that $n\ge 3$ and $\bE\in\rS_+(\Lambda^2(\C^n))$ is positive definite.
Let $f(\x,\y,\z)$ be defined as in \eqref{minprob}.  Then $(\x,\y,\z)\in\mathbb{P}(\C^n)^3$ is a nontrivial critical point of $f$ if $\x,\y,\z$ are linearly independent and the following conditions hold:
\begin{equation}\label{critptsf1}
\begin{aligned}
\frac{1}{\|\bE(\x\wedge\z)\|}\big(\bE^2(\x\wedge\z)\big)^\vee(\z)-\frac{1}{\|\bE(\x\wedge\y)\|}\big(\bE^2(\x\wedge\y)\big)^\vee(\y)=\lambda \x,\\
\frac{1}{\|\bE(\y\wedge\z)\|}\big(\bE^2(\y\wedge\z)\big)^\vee(\z)+\frac{1}{\|\bE(\x\wedge\y)\|}\big(\bE^2(\x\wedge\y)\big)^\vee(\x)=\mu \y,\\
-\frac{1}{\|\bE(\x\wedge\z)\|}\big(\bE^2(\x\wedge\z)\big)^\vee(\x)-\frac{1}{\|\bE(\y\wedge\z\|}\big(\bE^2(\y\wedge\z)\big)^\vee(\y)=\nu \z,
\end{aligned}
\end{equation}
where 
\begin{align*}
\lambda& =\|\bE(\x\wedge\z)\|-\|\bE(\x\wedge\y)\|,\\
\mu& =\|\bE(\y\wedge\z)\|-\|\bE(\x\wedge\y)\|,\\
\nu& =\|\bE(\x\wedge\z)\|+\|\bE(\y\wedge\z)\|.
\end{align*}
Suppose furthermore that $f(\x,\y,\z) < 0$.   Then $\z$ is orthogonal to $\x$ and $\y$, and either $\x$ is orthogonal to $\y$ or 
\begin{equation}\label{xyzfcond}
\|\bE(\x\wedge\z)\|-\|\bE(\y\wedge\z)\|=0.
\end{equation}
\end{lemma}
\begin{proof}
The proof of \eqref{critptsf1} is similar to the proof of Proposition \ref{propMxy}.  Introduce Lagrange multipliers $\lambda\langle \x,\x\rangle,\mu\langle \y,\y\rangle, \nu\langle \z,\z\rangle$ to deduce the first three equalities in \eqref{critptsf1}.   Now take inner product of the first equality with $\x$, the second equality with $\y$, the third equality with $\z$ to deduce the last three equalities.

Assume now that $f(\x,\y,\z) < 0$. Observe that $\lambda,\mu,-\nu<0$.
Take an inner product of the first identity with $\z$ and the last identity with $\x$.  Use  the equality \eqref{tildeAdef}, and the fact that $\z\wedge\z=\x\wedge\x=0$ to deduce the equalities
\begin{equation*}
\begin{aligned}
& \langle \bE(\x\wedge\y),\bE(\y\wedge \z)\rangle=\|\bE(\x\wedge\y)\|\lambda\langle\x,\z\rangle,\\
& \langle\bE(\y\wedge \z), \bE(\x\wedge\y)\rangle=-\|\bE(\y\wedge\z)\|\nu\langle\z,\x\rangle.
\end{aligned}
\end{equation*}
 
Assume that $\langle\x,\z\rangle\ne 0$.  Then we have the equality 
\begin{equation}\label{lamnueq}
\|\bE(\x\wedge\y)\|\lambda=-\|\bE(\y\wedge\z)\|\nu,
\end{equation}
Let 
\begin{equation*}
a=\|\bE(\x\wedge\z)\|,>0 \quad b=\|\bE(\y\wedge\z)\|>0, \quad c=\|\bE(\x\wedge\y)\|>0.
\end{equation*}
Then
\begin{equation*}
\lambda=a-c<0, \quad \mu=b-c<0, \quad \nu=a+b>0.
\end{equation*}
The equality \eqref{lamnueq} boils down to
\begin{equation*}
c(a-c)=-b(a+b)\Rightarrow \big(c-(a+b)\big)(c+b)=0\Rightarrow c-(a+b)=0\Rightarrow f(\x,\y,\z)=0,
\end{equation*}
which contradicts our assumption $f(\x,\y,\z)<0$.  Hence $\langle \x,\z\rangle=0$.
Similarly, $\langle \y,\z\rangle=0$.

Now, take the inner product of the first identity with $\y$ and the inner product of the second identity with $\x$ to obtain
\begin{equation*}
\begin{aligned}
\langle \bE(\x\wedge\z),\bE(\y\wedge \z)\rangle=\|\bE(\x\wedge\z)\|\lambda\langle\x,\y\rangle,\\
\langle\bE(\y\wedge \z), \bE(\x\wedge\z)\rangle=\|\bE(\y\wedge\z)\|\mu\langle\y,\x\rangle.
\end{aligned}
\end{equation*}
Assume that $\langle \x,\y\rangle\ne 0$.  Then
\begin{equation*}
a(a-c)=b(b-c)c\Rightarrow (a-b)(a+b-c)=0\Rightarrow a=b,
\end{equation*}
which proves \eqref{xyzfcond}.
\end{proof}
\subsection{The three dimenisonal case}\label{subsec:simpl} 
Recall (Corollary \ref{cor:E3dim}) that if $\dim \cH =~3$ then every operator $\bE \in \rB(\Lambda^2(\cH))$ is of the form \eqref{Euiujeig}. In order to simplify the notation let us assume that
\begin{equation}\label{eigE}
\begin{aligned}
& \bE(\bu_2\wedge\bu_3)=d_{23}\bu_2\wedge\bu_3, &&  \bE(\bu_3\wedge\bu_1)=d_{13} \bu_3\wedge\bu_1,\\
& \bE(\bu_1\wedge\bu_2)=d_{12} \bu_1\wedge\bu_2, && \text{with } \lambda_{23} \geq \lambda_{13} \geq \lambda_{12} >0.
\end{aligned}
\end{equation}



The following theorem is a refined version of Theorem \ref{necsufcondti} for $n=3$:
\begin{theorem}\label{minxyzthm}  Let $\cH$ be a three dimensional Hilbert space. 
Assume that $\bE\in \rS_+(\Lambda^2(\cH))$ and 
let $\mu(\bE)$ be given by \eqref{minprob}.
Then,
\begin{equation}\label{muEform}
\mu(\bE)=\min(0, d_{12}+d_{13}-d_{23}).
\end{equation}
Furthermore: 
\begin{enumerate} [(a)]
\item  
%
We have $\mu(\bE)=0$ if and only if one of the following conditions holds:

\begin{enumerate} [(i)]

\item Either $\x,\z$ or $\y,\z$ are colinear. If $d_{23}<d_{13}+d_{12}$ there are no other cases.

\item We have $d_{23}=d_{12}+d_{13}$ and the vectors $\x,\y,\z\in\mathbb{P}(\cH)$ are linearly independent, $\z=\bu_1$, and both $\x$ and $\y$ are orthogonal to $\z$.   Furthermore,
if $d_{12}=d_{13}$ then $\x=\bu_2,\y=\bu_3$ and if 
$d_{13}>d_{12}$ then  either $\{\x,\y\}=\{\bu_2,\bu_3\}$
or
\begin{equation}\label{xychoice}
\begin{aligned}
&\x=\sin\theta\bu_2+e^{\bi\psi}\cos\theta\bu_3 ,\;\y=-e^{-\bi\psi}\sin\phi\bu_2+\cos\phi \bu_3, \\
&\text{with } \theta,\phi\in(0,\pi/2),\psi\in[0,2\pi), \quad \tan\theta\tan\phi=\frac{d_{13}}{d_{12}}.
\end{aligned}
\end{equation}
\end{enumerate}
\item If $d_{23}>d_{12}+d_{13}$ then $\mu(\bE)=-d_{23}+d_{12}+d_{13}$.  The minimum is achieved if and only if  $\{\x,\y\}=\{\bu_2,\bu_3\}$ and $\z=\bu_1$.
\end{enumerate}
\end{theorem}
\begin{proof} Lemma \ref{wdge3dim} yields that there exists an orthonormal basis $\bu_1,\bu_2,\bu_3$  such that \eqref{eigE} holds.
We first show the inequality 
\begin{equation}\label{muEformin}
\mu(E)\le\min(0, d_{12}+d_{13}-d_{23}).
\end{equation}
Clearly, $f(\x,\x,\z)=0$, hence $\mu(\bE)\le 0$.
Assume that $d_{23}>d_{12}+d_{13}$.   Then $f(\bu_2,\bu_3,\bu_1)=d_{12}+d_{13}-d_{23}<0$. This proves \eqref{muEformin}.

Without loss of generality we assume that $\cH=\C^3$.   It is sufficient to study the cases $\x\wedge\z\ne \0, \y\wedge\z\ne \0$.  Since $\bE$ is positive definite,  Proposition~\ref{proplindep} yields that strict triangle inequality holds.
We thus restrict ourselves to the case where $\x,\y,\z$ are linearly independent, and $f(\x,\y,\z)=\mu(\bE)\le 0$ for some $\x,\y,\z\in \mathbb{P}(\C^3)$.
Lemma \ref{critptsf} yields that $\z$ is orthogonal to $\x$ and $\y$.  
 Without loss of generality we can assume $\z=\be_1=(1,0,0)^\top$.  
 Recall that the standard basis $\be_1,\be_2,\be_3\in\C^3$ induces a standard orthonormal basis 
 $\be_2\wedge\be_3, \be_3\wedge\be_1, \be_1\wedge\be_2\in \Lambda^2(\C^3)$.   
 Then, in this orthonormal basis, the operator $\bE^2$ is represented by a positive definite Hermitian matrix $H=[h_{ij}]\in\C^{3\times 3}$.  
 Let 
 $U=[u_{ij}]\in\C^{2\times 2}$ with $\det U=1$.  Change the basis $\be_1,\be_2,\be_3$ 
 to a new orthonormal basis 
 \begin{equation*}
 \bg_1=\be_1, \quad \bg_{i+1}=\sum_{j=1}^2 u_{ij}\be_{j+1} \textrm{ for }i\in[2].
 \end{equation*}
 Observe that $\bg_2\wedge\bg_3=\be_2\wedge\be_3$.  In this basis $\bE^2$ is represented by a Hermitian matrix $G=[g_{ij}]\in\C^{3\times 3}$.  Choose $U$ such that the $2\times 2$ principal submatrix $[g_{ij}]_{i,j\in\{2,3\}}$ is diagonal.  That is, $g_{23}=0$: 
\begin{equation*}
G=\begin{bmatrix}g_{11}&g_{12}&g_{13}\\g_{21}&g_{22}&0\\g_{31}&0&g_{33}\end{bmatrix}.
\end{equation*}
 Furthermore, we can assume that $g_{22}\ge g_{33}>0$.
 
As $\x,\y\in\mathbb{P}(\C^3)$ we can assume that $\x,\y$ are of the form:
 \begin{equation}\label{xzgform}
 \begin{aligned}
& \x=( \sin\theta ) \bg_2 + s (\cos\theta) \bg_3, \quad \y=-t (\sin\phi) \bg_2+ (\cos\phi) \bg_3, \\
& \text{with } \quad  \theta,\phi \in [0,\pi/2], \quad s,t\in\C, |s|=|t|=1.
\end{aligned} 
\end{equation}

 Then $\x\wedge\y=\big(\sin\theta\cos\phi+st\cos\theta\sin\phi\big)\bg_2\wedge\bg_3$ and 
 \begin{equation*}
 \begin{aligned}
 & \|\bE(\x\wedge\y)\|=\big|\sin\theta\cos\phi+st\cos\theta\sin\phi \big|\sqrt{g_{11}},\\
 & \|\bE(\x\wedge\z)\|=\sqrt{g_{22}\cos^2\theta+g_{33}\sin^2\theta},\\
 & \|\bE(\y\wedge\z)\|=\sqrt{g_{22}\cos^2\phi+g_{33}\sin^2\phi}.
 \end{aligned}
 \end{equation*}
 Hence,
\begin{multline*}
f(\x,\y,\z)=\|\bE(\x\wedge\z)\|+\|\bE(\y\wedge\z)\|-\|\bE(\x\wedge\y)\|=\\
\;\; = \sqrt{g_{22}\cos^2\theta+g_{33}\sin^2\theta}\,+\sqrt{g_{22}\cos^2\phi+g_{33}\sin^2\phi}\,-\\
- |\sin\theta\cos\phi+st\cos\theta\sin\phi|\sqrt{g_{11}}.
 \end{multline*}
 Since we assumed that $\mu(\bE)=f(\x,\y,\z)$ we must have the equality
 \begin{equation*}
 |\sin\theta\cos\phi+st\cos\theta\sin\phi|=\sin\theta\cos\phi+\cos\theta\sin\phi
 \Rightarrow s=e^{\bi\psi}, t=e^{-\bi\psi}.
 \end{equation*}
 
 As we assumed that $g_{23}=0$ it follows that the eigenvalues of the submatrix $[g_{ij}]_{i=j=2}^3$ are $g_{22}\ge g_{33}$.  The Cauchy interlacing theorem and the maximum characterization of $d_{23}^2$ 
 \begin{equation*}
 g_{22}\ge d_{13}^2\ge g_{33}\ge d_{12}^2, \quad g_{11}\le d_{23}^2.
 \end{equation*}
Hence, employing the function $F$ defined in \eqref{ineqfort1i}, we have
\begin{equation}\label{fFineq}
f(\x,\y,\z)\ge F(d_{13},d_{12},t,\theta,\phi), \textrm{ for }t=\frac{d_{23}}{d_{13}+d_{12} }.
\end{equation}

We first consider the case $d_{23}<d_{12}+d_{13}$.  Then $t<1$ and
Lemma \ref{ineqforti} implies that $f(\x,\y,\z)> 0$ contrary to our assumption that $f(\x,\y,\z)=\mu(\bE)\le 0$.  Hence, in this case $\mu(\bE)=0$.  That is,  the triangle inequality holds, and equality in the triangle  inequality holds if either $\x\wedge\z=\0$ or $\y\wedge \z=\0$.

Assume now that $d_{23}= d_{12}+d_{13}$.  Note that $f(\x,\y,\z)\le d_{12}+d_{13}-d_{23}=0$.  Next observe that $t$ in the inequality \eqref{fFineq} satisfies $t\le 1$.  Lemma \ref{ineqforti} then implies that $f(\x,\y,\z)\ge 0$.
Therefore $f(\x,\y,\z)=\mu(\bE)=0$ in this case too.   That is, the triangle inequality holds.

We now consider the equality case in the triangle inequality, that is $f(\x,\y,\z)=0$. Consequently,  $g_{11}=d_{23}^2=(d_{12}+d_{13})^2$.   In this case we claim that $G$ is a diagonal matrix $\diag(d_{23}^2, d_{13}^2,d_{12}^2)$.
Indeed, because  
$$g_{11}=\langle \bE^2(\bg_2\wedge\bg_3),\bg_2\wedge\bg_3\rangle=d_{23}^2,$$
and $d_{23}$ is the maximal eigenvalue of $\bE$, we must have 
$\bg_2\wedge\bg_3=\zeta \bu_2\wedge\bu_3$ for some $|\zeta|=1$. Hence the entries $(1,i),(i,1)$ of $G$ are zero. for $i\in[2]$.  Since we assumed that $g_{23}=0$ we deduce that $G$ is diagonal.   Therefore, $\bg_1=\bu_1=\bz$, and 
\begin{equation}\label{Geigeq}
\begin{aligned}
&\bE^2(\bg_2\wedge\bg_3)=d_{23}^2 \bg_2\wedge\bg_3,\quad\bE^2(\bg_3\wedge\bg_1)=d_{13}^2\bg_3\wedge\bg_1,\\
&\bE^2(\bg_1\wedge\bg_2)=d_{12}^2\bg_1\wedge\bg_2.
\end{aligned}
\end{equation}

Suppose first that $d_{12}=d_{13}$.    Lemma \ref{ineqforti} implies that  $F(d_{12},d_{12},1,\theta,\phi)=0$ if and only if $\theta+\phi=\pi/2$.  Recall that in this case equality \eqref{xychoice} holds and the equalities.  Use equalities \eqref{xzgform} and $s=e^{\bi\psi}, t=e^{-\bi\psi}$ to deduce 
$\langle\x,\y\rangle=0$.  As $d_{12}=d_{13}$ it follows that we can assume that $\bu_2=\x,\bu_3=\y$.

Assume now that $d_{13}>d_{12}$.  Then $\bg_2=\bu_2, \bg_3=\bu_3$.
Lemma \ref{ineqforti} yields that $F(d_{13},d_{12},1,\theta,\phi)=0$ if and only if \eqref{F=0cond} hold.  The first condition in \eqref{F=0cond} yields that $\{\x,\y\}=\{\bu_2,\bu_3\}$. The second condition of\eqref{F=0cond} is the condition \eqref{xychoice}.

We now assume that $d_{23}>d_{12}+d_{13}$.  Then $f(\x,\y,\z)=\mu(\bE)<0$ and $t>1$ in \eqref{fFineq}.  Lemma \ref{ineqforti} yields that $\omega\big(d_{23},d_{13},d_{23}/(d_{13}+d_{12}) \big)=-d_{23} +d_{13}+d_{12}$.  For $f(\x,\y,\z)=\omega \big(d_{23},d_{13},d_{23}/(d_{13}+d_{12}) \big)$ we must have that 
$g_{11}=d_{23}^2$.  Therefore, $G=\diag(d_{23}^2,d_{13}^2,d_{12}^2)$, and \eqref{Geigeq} holds.  Furthermore, the first condition of \eqref{F=0cond} holds.
Hence $\{\x,\y\}=\{\bg_2,\bg_3\}, \z=\bg_1$.
\end{proof}


\label{proof_thm41}{\it Proof of Theorem \ref{necsufcondti}.}  Assume that $\x,\y,\z\in\mathbb{P}(\C^n)$ are linearly dependent.   Proposition \ref{proplindep} yields that the triangle inequality holds.  Assume that $\x,\y,\z\in\mathbb{P}(\C^n)$ are linearly independent.  Let $\W= \text{span}(\x,\y,\z)$.  Denote by $\bF(\W)$ the restriction of $\bE^2$ to $\Lambda^2(\W)$.
Denote $\bM=\sqrt{\bF(\W)}$.  Let $d_{23}\ge d_{13}\ge d_{12}\ge 0$ be the eigenvalues of $\bM$ with the corresponding  orthonormal eigenvectors $\bu_2\wedge\bu_3,\bu_3\wedge\bu_1, \bu_1\wedge\bu_2$.  Assume that the triangle inequality holds.  Hence $d_{12}+d_{13}-d_{23}\ge 0$.  Vice versa, assume that  $d_{12}+d_{13}-d_{23}\ge 0$.
Let $\varepsilon>0$ and denote by $M_{\varepsilon}$ a positive definite operator in $\rS_+(\Lambda^2(\W))$ whose eigenvalues are $d_{23}+2\varepsilon,d_{12}+\varepsilon,d_{13}+\varepsilon$ with the same corresponding eigenvectors as $\bM$.  Theorem \ref{minxyzthm} yields that for $\bE=\bM_{\varepsilon}$ the triangle inequality holds.  Let $\varepsilon\searrow 0$ to deduce that for $\bE=\bM$  the triangle inequality holds.\qed

The proof of Theorem \ref{necsufcondti} yields and important simplification of the triangle inequality \eqref{triangins}.
\begin{corollary}\label{cornecsufcondti}  Assume that $\dim\cH\ge 3$ and $\bE\in\rS_+(\Lambda^2(\cH))$.  Then the triangle inequality \eqref{triangins} holds if and only if  it holds for all orthonormal triples of vectors $\x,\y,\z \in \cH$.
\end{corollary}
\subsection{A simple sufficient condition}\label{subsec:sufto}
We now give a simple sufficient condition for $\bE$ to be a triangular operator.
\begin{proposition}\label{sufcond}  Assume that $\dim\cH = n \ge 3$, set $m = \binom{n}{2} = \dim \Lambda^2(\cH)$, and let $\bE\in\rS_+(\Lambda^2(\cH))$ have eigenvalues:
\begin{equation*}
\begin{aligned}
\lambda_1(\bE)\ge \ldots\ge \lambda_{m -1}(\bE)\ge \lambda_{m }(\bE)>0.
\end{aligned}
\end{equation*}
If
\begin{equation}\label{sufcond1} 
\lambda_{m -1}(\bE)+\lambda_{m }(\bE)\ge \lambda_1(\bE),
\end{equation}
then $\bE\in\cT(\cH)$.
\end{proposition}
\begin{proof}
We will first show that if inequality \eqref{sufcond1} holds for $m(n) > 3$ then, for any $(n-1)$-dimensional subspace $\W \subseteq \cH$,
\begin{equation}\label{lamFWineq}
\sqrt{\lambda_{m' -1}(\bF(\W))} +\sqrt{\lambda_{m'}(\bF(\W))}\ge\sqrt{\lambda_1(\bF(\W))}\,,
\end{equation}
with $\bF(\W)$ denoting the restriction of $\bE^2$ to $\Lambda^2(\W)$ and $m' = \dim \Lambda^2(\W) = \binom{n-1}{2} = m - n +1$.


Let $\bv_1,\ldots,\bv_{n-1},\bv_n$ be an orthonormal basis of $\cH$ such that 
$\bv_1,\ldots,\bv_{n-1}$ is an othonormal basis of $\W$.  Set
\begin{equation*} 
\V_1=\Lambda^2(\W),\,\V_{i+1}=\textrm{span}(\V_1,\bv_1\wedge\bv_n,\ldots,\bv_{i}\wedge\bv_n), \quad i\in[n-1],
\end{equation*}
to be a partial flag of subspaces in $\Lambda^2(\cH)$.  Denote by $\bF_i$ the restriction of $\bE^2$ to $\V_i$.  Observe that $\bF_1=\bF(\W)$,  $\bF_n=\bE^2$.
Recall that the eigenvalues of $\bF_i$ interlace with the eigenvalues of $\bF_{i+1}$, 
\begin{equation*}
\lambda_j(\bF_{i+1})\ge \lambda_j(\bF_{i}))\ge \lambda_{j+1}(\bF_{i+1}), \quad \text{for} \quad j\in[m' +i-1], \, i\in[n-1].
\end{equation*}
In particular, for all $i\in[n-1]$,
\begin{equation}\label{intFiineq}
\begin{aligned}
&\sqrt{\lambda_1(\bF_{i+1})}\ge \sqrt{\lambda_1(\bF_{i})},\\
&\sqrt{\lambda_{m' +i-2}(\bF_{i})}\ge \sqrt{\lambda_{m' +i-1}(\bF_{i+1})},\\
&\sqrt{\lambda_{m' +i-1}(\bF_{i})}\ge \sqrt{\lambda_{m' +i}(\bF_{i+1})}.
\end{aligned}
\end{equation}
Assume the inequality \eqref{sufcond1} for $\bF_n=\bE^2$.  Use the above inequality for $i=n-1$ to deduce
\begin{equation*}
\sqrt{\lambda_{m'+i-1}(\bF_{i})} +\sqrt{\lambda_{m'+1}(\bF_i)}\ge\sqrt{\lambda_1(\bF_i)}.
\end{equation*}
Continue by induction on $i$ to deduce \eqref{lamFWineq} for $i=1$.

Next, we proceed by induction on $n$ to deduce that if \eqref{sufcond1} holds for some $n\geq3$ then \eqref{lamFWineq} is true for any 3-dimensional subspace $\W \subseteq \cH$. Then, Theorem \ref{necsufcondti} implies that the operator $\bE$ is triangular.
\end{proof}

\section{The 2-cone of triangular operators}\label{sec:to}
Let
\begin{equation}\label{defS1S+1}
\begin{aligned}
\rS_1(\cH)=\{S\in\rS(\cH), \tr S=1\}, &&
\rS_{+,1}(\cH)=\rS_+(\cH)\cap\rS_1(\cH).
\end{aligned}
\end{equation}
Observe that $\rS_{+,1}(\cH)$ is a compact convex set of density operators on $\cH$ of real dimension $\dim^2\cH-1$.
\begin{definition}\label{def2-cone}
Let $\cH$ be a finite dimensional Hilbert space.  A set $\cC\subseteq\rS_+(\cH)$ is called
a 2-cone if the following conditions hold:
\begin{enumerate} [(a)]
\item $0\in\cC$.
\item If $T\in\cC$ then $tT\in\cC$ for $t\in[0,\infty)$.
\item If $S,T\in \cC$ then $\sqrt{S^2+T^2}\in\cC$.
\end{enumerate} 
A 2-cone $\cC$ is trivial if $\cC=\{0\}$.
\end{definition}
For a subset $\cS\subseteq \rS_+(\cH)$ and $p>0$ denote
\begin{equation}\label{defSp}
\cS^{p}=\{T^p, \, T\in\cS\} \subseteq \rS_+(\cH).
\end{equation}
The following lemma is straightforward:
\begin{lemma}\label{equiv2con}  A set $\cC\subseteq \rS_+(\cH)$ is a 2-cone if and only if
$\cC^2$ is a cone in $\rS_+(\cH)$.
\end{lemma}

An element $T\in\cC\setminus\{0\}$ is called and an extreme ray if and only if $T^2$ is an extreme ray in $\cC^2$.  That is, if $T^2=T_1^2+T_2^2$ for some $T_1,T_2\in \cC$ then $T_1$ and $T_2$ are colinear.  Denote by Ext$(\cC)$ the set of extreme rays of $\cC$.
Observe that the set of extreme rays of the cone $\cC^2$ is determined by $\text{Ext}^2(\cC)=\text{Ext}(\cC^2)$.
\begin{proposition}\label{extremrays}
Assume that $\dim\cH=n$ and let $\cC\subset \rS_+(\cH)$ be a closed nontrivial 2-cone.
Then $\mathrm{Ext}(\cC)\ne\emptyset$ and every $T^2\in\cS^2$ is the sum of at most
$n^2$ of squares of distinct extreme rays of $\cS$.
\end{proposition}
\begin{proof} 
Let
\begin{equation*}
\cC_1=\cC\cap (\rS_{+,1}(\cH))^{1/2}.
\end{equation*} 
Then $\cC_1^2$ is a compact convex  set in $\rS_{+,1}(\cH)$.  Carath\'eodory's theorem yields that Ext$(\cC_1)$, the set of the extreme points of $\cC_1$, is nonempty, and every point of $\cC_1$ is a convex combination of at most $\dim\cC_1+1$ distinct
extreme points of $\cC_1$ \cite{Roc70}.  As $\dim \cC_1\le \dim\rS_{+,1}(\cH)=n^2-1$ we deduce the proposition.
\end{proof}

\begin{definition}\label{triangop}
Let $\cH$ be a finite dimensional  Hilbert space of $\dim \cH\ge 2$.  An operator $\bE
\in \rS_+(\Lambda^2(\cH))$ is called \emph{triangular} if the inequality \eqref{triangins} holds.
Denote by $\cT(\cH)\subset  \rS_+(\Lambda^2(\cH))$ the set of triangular operators.
\end{definition}

\begin{proposition}\label{cT2cone}
Assume that $\cH$ is a Hilbert space with $\dim\cH\ge 2$.   Then $\cT(\cH)$ is a closed $2$-cone.
\end{proposition}
\begin{proof} Clearly, $\cT(\cH)$ is a closed set.   Furthermore, if $\bE\in\cT(\cH)$ then
$e\bE\in\cT(\cH)$ for $e\ge 0$.
Suppose that $\bE_1,\bE_2\in\cT(\cH)$.    For $i\in[2]$ set
\begin{equation*}
\begin{aligned}
&a_i=\|\bE_i(\x\wedge\z)\|=\sqrt{\langle\bE^2_i(\x\wedge\z),\x\wedge\z\rangle},\\
&b_i=\|\bE_i(\y\wedge\z)\|=\sqrt{\langle\bE^2_i(\y\wedge\z),\x\wedge\z\rangle},\\
&c_i=\|\bE_i(\x\wedge\y)\|=\sqrt{\langle\bE^2_i(\x\wedge\y),\x\wedge\y\rangle},\\
&\ba=(a_1,a_2)^\top, \bb=(b_1,b_2)^\top,  \bc=(c_1,c_2)^\top\in\R_+^2, \,\|\x\|=\|\y\|=\|\z\|=1.
\end{aligned}
\end{equation*}
Observe that $\ba+\bb\ge \bc\ge \0$.  Hence,  $\|\ba+\bb\|\ge \|\bc\|$.
Recall the triangle inequality $ \|\ba\|+\|\bb\|\ge \|\ba+\bb\|$.  Hence,  $\|\ba\|+\|\bb\|\ge \|\bc\|$, which is equivalent to 
\begin{equation*}
\begin{aligned}
&\sqrt{\langle(\bE^2_1+\bE_2^2)(\x\wedge\z),\x\wedge\z\rangle}+
\sqrt{\langle(\bE^2_1+\bE_2^2)(\y\wedge\z),\y\wedge\z\rangle}\ge\\
&\sqrt{\langle(\bE^2_1+\bE_2^2)(\x\wedge\y),\x\wedge\y\rangle}.
\end{aligned}
\end{equation*}
Set $\bE=\sqrt{\bE_1^2+\bE_2^2}$ to deduce that $\bE\in\cT(\cH)$.
\end{proof}

\begin{proposition}\label{prop:sym}
The unitary group $\cU = \{ U \wedge U, \, U \in \rU(\cH)\} \subset B(\Lambda^2(\cH))$ is a symmetry of the 2-cone of triangular operators, i.e. $\cU^* \cT(\cH) \cU = \cT(\cH)$. 
\end{proposition}
\begin{proof} 
Let $\bE \in \cT(\cH)$ and $U \wedge U \in \cU$. For any $\x,\y \in \cH$ we have 
\begin{align*}
\big\| \big[ (U \wedge U)^* \, \bE\, (U \wedge U) \big] (\x\wedge\y) \big\| & =\sqrt{\langle\bE^2(U \x\wedge U\y),U\x\wedge U\y\rangle}\\
& = \big\| \bE\, (U\x\wedge U\y) \big\|.
\end{align*}
Hence, the operator $(U \wedge U)^* \, \bE\, (U \wedge U)$ satisfies the triangle inequality \eqref{triangins} if and only if $\bE$ does.
\end{proof}

\subsection{Extreme rays of $\cT(\C^3)$}\label{subsec:ern=3}
One way to characterize $\cT(\cH)$ is through its extreme rays.
It is enough to characterize the extreme points of $\cT_1(\cH)$ --- the set of triangular operators of trace one.

\begin{theorem}\label{ern=3} Assume that $\dim\cH=3$.  An element $\bA \in \cT(\cH)$ is an extreme ray if and only if the equality
\begin{equation}\label{lam1=lam2+lam3}
d_{23}(\bA)=d_{12}(\bA)+d_{13}(\bA)>0.
\end{equation}
\end{theorem}
\begin{proof}   Theorem \ref{necsufcondti} asserts that 
$\bE\in\cT(\cH)\setminus\{0\}$ if and only if 
\begin{align}\label{dt}
0<d_{23}(\bE)\le d_{12}(\bE)+d_{13}(\bE).
\end{align}
Recall that $\bE$ has three orthonormal vectors that are given by \eqref{eigE}.
Assume first that the second inequality in \eqref{dt} is strict.  Then $d_{12}(\bE)>0$ 
and there exists $\varepsilon>0$ such that 
\begin{equation*}
\begin{aligned}
& \sqrt{d_{23}^2(\bE)+\varepsilon}<d_{12}(\bE)+d_{13}(\bE), \qquad 0<d_{23}^2(\bE)-\varepsilon,\\
& \max \left(\sqrt{d_{23}^2(\bE)-\varepsilon},d_{13}(\bE)\right)< \min \left(\sqrt{d_{23}^2(\bE)-\varepsilon},d_{13}(\bE)\right) + d_{12}(\bE).
\end{aligned}
\end{equation*}
Define $\bE_1,\bE_2$ as follows:
\begin{equation*}
\begin{aligned}
\bE_1(\bu_2\wedge\bu_3)=\sqrt{d_{23}(\bE)+\varepsilon}\,\bu_2\wedge\bu_3, \quad 
\bE_2(\bu_2\wedge\bu_3)=\sqrt{d_{23}(\bE)-\varepsilon}\,\bu_2\wedge\bu_3,
\end{aligned}
\end{equation*}
and $\bu_3\wedge\bu_1$ and $\bu_1\wedge\bu_2$ are the eigenvectors of $\bE_1,\bE_2$ with the same eigenvalues as $\bE$.
Theorem  \ref{necsufcondti} yields that $\bE_1,\bE_2\in\cT(\cH)$.   Clearly $\bE^2=\big(\frac{1}{\sqrt{2}}\bE_1\big)^2+\big(\frac{1}{\sqrt{2}}\bE_2\big)^2$.  Hence $\bE$ is not an extreme ray.

Assume now that \eqref{lam1=lam2+lam3} holds for $\bA=\bE$.  Suppose to the contrary that $\bE$ is not an extreme ray.  Without loss of generality we can assume that $\bE\in\cT_1(\cH)$.  Thus, $\bE^2$ is not an extreme point in  the compact closed set $\cT_1^2(\cH)$.
Hence, there exists $2\le k\le 9$ such that 
\begin{equation*}
\bE^2=\sum_{i=1}^k a_i\bA_i^2, \text{ with } \bA_i\in \cT_1(\cH),a_i>0, {\textstyle\sum_i} a_i=1, \text{ and } \bA_i\ne\bA_j \textrm{ for } i\ne j.
\end{equation*}
Now, observe that 
\begin{equation*}
\bF(\varepsilon)=\bE^2\pm \varepsilon \bB \in \cT_1^2(\cH),\textrm{ for }\bB=t(\bA_1^2-\bA_2^2), t=\min(a_1,a_2),\varepsilon\in[-1,1].
\end{equation*}
We claim that for small enough positive $\varepsilon$, the operators  $\bF(\varepsilon)$ and $\bF(-\varepsilon)$ are not in $\cT_1^2(\cH)$.

Choose an orthonormal basis $\bu_2\wedge\bu_3,\bu_3\wedge\bu_1, \bu_2\wedge\bu_1$ in $\Lambda^2(\cH)$ that satisfy \eqref{eigE}.  In this basis 
$\bE^2$ and $\bB$ have the form
\begin{equation*}
\bE^2=\begin{bmatrix}\lambda_1^2&0&0\\0&\lambda_{2}^2&0\\0&0&\lambda_{3}^2\end{bmatrix},\,\bB=\begin{bmatrix}b_{11}&b_{12}&b_{13}\\ \bar b_{12}&b_{22}&b_{23}\\\bar b_{13}&\bar b_{23}&b_{33}\end{bmatrix}\ne 0, b_{11}+b_{22}+b_{33}=0.
\end{equation*}
Were we set $\lambda_1 = d_{23}, \lambda_2 = d_{13}, \lambda_{3} = d_{12}$ for more convenient notation in the coming calculations.
As $\bF(\pm 1)$ is positive semidefinite it follows that $b_{13}=b_{23}=b_{33}=0$ if $\lambda_{3}=0$.   Recall Rellich's theorem \cite{Frib}, which implies that the eigenvalues of $\bF(\zeta)$
are analytic functions of $\zeta$ in some open simply connected domain containing $\R$.  In particular,  the eigenvalues of $\bF(\zeta)$ are given by three convergent Taylor series
\begin{equation*}
\nu_i(\zeta)=\sum_{l=0}^\infty \nu_{l,i}\zeta^l, \quad |\zeta|<r, \, i\in[3], r\in(0,1)
\end{equation*}
As $\tr \bB=0$ we deduce
\begin{equation}\label{trB=0}
\sum_{i=1}^3 \nu_{l,i}=0 \textrm{ for } l\in\N.
\end{equation}

Observe that the assumption that $\bF(\zeta)\in\cT_1^2(\cH)$ for $\zeta\in(-r,r)$ is
\begin{equation}\label{Fzetasum}
{\textstyle\sqrt{\nu_j(\zeta)}\le \sqrt{\nu_k(\zeta)}+\sqrt{\nu_l(\zeta)}},\quad \text{ for } \quad j,k,l\in[3], \zeta\in(-r,r).
\end{equation}
Suppose first that $\lambda_{2}>\lambda_{3}>0$.  Then $\lambda_{1}=\lambda_{2}+\lambda_{3}>\lambda_{2}$.  Therefore, for small enough $r>0$ we have 
\begin{equation*}
\lambda_1(\bF(\zeta))=\nu_1(\zeta)>\lambda_{2}(\bF(\zeta))=\nu_2(\zeta)>
\lambda_{3}(\bF(\zeta))=\nu_3(\zeta)>0,\,\zeta\in(-r,r).
\end{equation*}
Recall the formulas for $\nu_{1,i}$ \cite[\textsection 3.8]{Frib} to deduce
\begin{equation}\label{firstvar}
\nu_i(\zeta)=\lambda_i^2+b_{ii}\zeta+O(\zeta^2),  \quad \nu_{1,i}=b_{ii}, \quad i\in[3].
\end{equation}
Hence,
\begin{equation*}
\sqrt{\nu_i(\zeta)}=\lambda_i+\frac{b_{ii}\zeta}{2\lambda_i} +O(\zeta^2), 
\end{equation*}
Inequality \eqref{Fzetasum} for $j=1, k=2,l=3$ yields
\begin{equation*}
\frac{b_{11}\zeta}{2\lambda_1}\le \frac{b_{22}\zeta}{2\lambda_2} +
\frac{b_{33}\zeta}{2\lambda_3} \textrm{ for } \zeta\in(-r,r).
\end{equation*}
Hence,
\begin{equation}\label{zerofvar}
\frac{b_{11}}{2\lambda_1}=\frac{b_{22}}{2\lambda_2} +
\frac{b_{33}}{2\lambda_3}.
\end{equation}
Recall the equalities $b_{11}+b_{22}+b_{33}=0$ and \eqref{lam1=lam2+lam3}, which
yield  
\begin{equation}\label{biieq}
b_{11}=-(b_{22}+b_{33}), \quad \lambda_3(\lambda_3+2\lambda_3)b_{22}=-\lambda_2(\lambda_2+2\lambda_3)b_{33}.
\end{equation}
We now recall the formulas for $\nu_{i,2}$ in \cite[(4.20.2)]{Frib}:
\begin{equation*}
\nu_{i,2}=\be_i^* \bB^*(\lambda_i^2\bI-\bE^2)^\dagger\bB\be_i, \quad \be_i=(\delta_{1i},\delta_{2,i},\delta_{3i})^\top, \quad i\in[3].
\end{equation*}
Here $\bI\in\R^{3\times 3}$ is the identity matrix and $(\lambda_i^2\bI-\bE^2)^\dagger$ is the Moore--Penrose inverse.  Hence:
\begin{equation}\label{secondvar}
\begin{aligned}
& \nu_{1,2}=\frac{|b_{12}|^2}{\lambda_1-\lambda_2}+\frac{|b_{13}|^2}{\lambda_1-\lambda_3}=\frac{|b_{12}|^2}{\lambda_3}+\frac{|b_{13}|^2}{\lambda_2},\\
& \nu_{2,2}=-\frac{|b_{12}|^2}{\lambda_1-\lambda_2}+\frac{|b_{23}|^2}{\lambda_2-\lambda_3}=-\frac{|b_{12}|^2}{\lambda_3}+\frac{|b_{23}|^2}{\lambda_2-\lambda_3},\\
& \nu_{3,2}=-\frac{|b_{13}|^2}{\lambda_1-\lambda_3}+\frac{|b_{13}|^2}{\lambda_1-\lambda_3}=-\frac{|b_{13}|^2}{\lambda_2}-\frac{|b_{23}|^2}{\lambda_2-\lambda_3}.
\end{aligned}
\end{equation}
Note that the above $\nu$'s satisfy equality \eqref{trB=0} for $l=2$.
We now compute the Taylor expansions of $\sqrt{\lambda_i(\bF(\zeta))}$ up to the $\zeta^2$:
\begin{equation*}
\sqrt{\lambda_i(\bF(\zeta))}=\lambda_i \left(1+\frac{1}{2\lambda_i^2}\nu_{1,i}\zeta +\frac{1}{2\lambda_i^2}\nu_{2,i}\zeta^2-\frac{1}{8\lambda_i^4}\nu_{1,i}^2\zeta^2 \right)+O(\zeta^3).
\end{equation*}
We claim that:
\begin{multline}\label{difeigexp}
\sqrt{\lambda_1(\bF(\zeta))} -\sqrt{\lambda_2(\bF(\zeta))}-\sqrt{\lambda_3(\bF(\zeta))}\\
 = \zeta^2\left(\frac{\nu_{1,2}}{2\lambda_1}-\frac{\nu_{1,1}^2}{8\lambda_1^3}
-\frac{\nu_{2,2}}{2\lambda_2}+\frac{\nu_{2,1}^2}{8\lambda_2^3}
-\frac{\nu_{3,2}}{2\lambda_2}+\frac{\nu_{3,1}^2}{8\lambda_3^3}\right)+O(\zeta^3).
\end{multline}
 The equality $\lambda_1=\lambda_2+\lambda_3$ yields 
that first coefficient of the above Taylor expansion is zero.  The equality \eqref{zerofvar} yields that the second coeffcient  of the above Taylor expansion is zero.  We claim that the coefficient of $\zeta^2$ is positive.

Consider first
\begin{multline*}
\frac{\nu_{1,2}}{\lambda_1}-\frac{\nu_{2,2}}{\lambda_2} -\frac{\nu_{3,2}}{\lambda_3}\\
= \frac{|b_{12}|^2}{\lambda_1\lambda_3}+\frac{|b_{13}|^2}{\lambda_1\lambda_2} +\frac{|b_{12}|^2}{\lambda_2\lambda_3}- \frac{|b_{23}|^2}{\lambda_2(\lambda_2-\lambda_3)} +\frac{|b_{13}|^2}{\lambda_2\lambda_3}+\frac{|b_{23}|^2}{\lambda_3(\lambda_2-\lambda_3)}.
\end{multline*}
As we assume that $\lambda_2>\lambda_3$ the above expression is positive, 
unless $b_{12}=b_{13}=b_{23}=0$.

Next we consider the expression:
\begin{equation*}
\begin{aligned}
-\frac{\nu_{1,1}^2}{\lambda_1^3} +\frac{\nu_{1,1}^2}{\lambda_2^3}+\frac{\nu_{3,1}^2}{\lambda_3^3}=-\frac{(b_{22}+b_{33})^2}{\lambda_1^3} +\frac{b_{22}^2}{\lambda_2^3}+\frac{b_{33}^2}{\lambda_3^3}.
\end{aligned}
\end{equation*}
In view of the last equality in \eqref{biieq} we deduce that if $(b_{22},b_{33})^\top\ne \0$ then $b_{22}b_{33}<0$.  As $\lambda_1>\lambda_2>\lambda_3$ we deduce that the  above expression is positive, unless $b_{11}=b_{22}=b_{33}=0$.  As $\bB\ne 0$ we must have that the coefficient of $\zeta^2$ in the expression \eqref{difeigexp} is positive.  Hence for small nonzero $\zeta$ the operator $\bF(\zeta)$ is not in $\cT_1^2(\cH)$, contrary to our assumptions.

We now consider the case when $\lambda_2=\lambda_3>0$ and $\lambda_1=2\lambda_2$.   Let $U\in\C^{3\times 3}$ be a unitary matrix of the form
$\begin{bmatrix}1&0&0\\0&u_{22}&u_{23}\\0&u_{32}&u_{33}\end{bmatrix}$.
Note that $U^*\bE^2U=\bE^2$. Now choose $U$ such that the matrix $U^*\bB U=\bC=[c_{ij}]$, such that $c_{23}=0$.  Thus, without loss of generality we can assume that $\bB=\bC$, and $b_{23}=0$.  
Then $\lambda_1(\bF(\zeta))=\nu_1(\zeta)$ and we have the same Taylor expansion as before, with $\lambda_2=\lambda_3$,  and equalities \eqref{firstvar} hold.
We claim that equalities \eqref{secondvar} hold, where first we let the term
$\frac{|b_{23}|^2}{\lambda_2-\lambda_3}=0$ and then set $\lambda_2=\lambda_3$.
This equalities can be deduced as follows.  Assume that $\lambda_1=\lambda_2+\lambda_3$, and $b_{23}=0$.  Then we have equalities \eqref{secondvar} with $b_{23}=0$.  Let $\lambda_2\searrow\lambda_3$.
Use the same arguments as before to deduce the contradiction.

We are left with the case $\lambda_3=0$, $\lambda_1=\lambda_2>0$.
As we pointed out before we must have $b_{33}=0$, hence $b_{13}=b_{23}=0$.   Use a unitary matrix $\begin{bmatrix}u_{11}&u_{12}&0\\u_{21}&u_{22}&0\\ 0&0&0\end{bmatrix}$ to deduce that we can assume that $\bB$ is a diagonal matrix $\diag(b_{11},b_{22},0)$ where $b_{11}=-b_{22}\ne 0$.  Then 
\begin{equation*}
\lambda_1(\bF(\zeta))=\lambda_1^2+|b_{11}||\zeta|> \lambda_2(\bF(\zeta))=\lambda_2^2-|b_{11}||\zeta| \textrm{ for } \zeta\in(0,r).
\end{equation*}
Again,  $\bF(\zeta))\not \in\cT_1^2(\cH)$ contradictory to our assumptions.
\end{proof}
\section{Special subclasses of triangular operators 
}\label{sec:staoI}

\subsection{Distance-induced triangular operators}
In this section we consider a special subset of operators $\bE$, which are of the form \eqref{Euiujeig}. Recall that the eigenvalues of such operators can be organised into a symmetric matrix with zero diagonal, $D(\bE)=[d_{ij}]\in\R^{n\times n}_+$. Then, using the basis $\{\bu_i\}_i$ of $\cH$, we can write
\begin{equation}\label{d(x,y)form}
\begin{aligned}
&d_\bE(\x,\y)=\|\bE(\x\wedge\y)\|=\sqrt{\frac{1}{2}\sum_{i,j=1}^n d_{ij}^2 |x_iy_j-x_jy_i|^2},\\
&\text{where } \quad \x=\sum_{i=1}^nx_i\bu_i, \quad y=\sum_{i=1}^n y_i\bu_i.
\end{aligned}
\end{equation}
Semi-distances \eqref{d(x,y)form} are known under the name of quantum 2-Wasserstein semi-distances. They arise from the quantum optimal transport problem as follows (see \cite{CEFZ21,FECZ21} for the details):

Let $\rho^A,\rho^B \in \S_{+,1}(\cH)$ be two density operators on a Hilbert space $\cH = \C^n$ and let $\Gamma^Q(\rho^A,\rho^B) = \big\{ \rho^{AB} \in  \S_{+,1}(\cH \otimes \cH), \, \tr_A \rho^{AB} = \rho^B, \tr_B \rho^{AB} = \rho^A \big\}$ be the set of their ``quantum couplings''. For a matrix $D(\bE)$ as above define the corresponding \emph{quantum cost matrix}, $C_\bE \in \S_{+}(\cH \otimes \cH)$, as
\begin{align*}
    C_\bE \vc \sum_{j>i=1}^n d_{ij} \, (\bu_i \wedge \bu_j)(\bu_i \wedge \bu_j)^*.
\end{align*}
The quantum 2-Wasserstein semi-distance on $\S_{+,1}(\cH)$ associated with such a $C_\bE$ is defined as
\begin{align*}
    \mathrm{W} (\rho^A,\rho^B) \vc \sqrt{ \min_{\rho^{AB} \in \Gamma^Q(\rho^A,\rho^B)} \tr \big( \rho^{AB} \, C_\bE^2 \big) }.
\end{align*}
If either of the states $\rho^A,\rho^B$ is pure, then $\Gamma^Q(\rho^A,\rho^B) = \{\rho^A \otimes \rho^B\}$. Consequently, if $\rho^A = \x \x^*$ and $\rho^A = \y \y^*$, for two vectors $\x, \y \in \mathbb{P}(\C^n)$, then
\begin{align*}
    \mathrm{W} (\rho^A,\rho^B) = \sqrt{  \sum_{j>i=1}^n d_{ij}^2 \, \Tr \big[ (\x \x^* \otimes \y \y^*)(\bu_i \wedge \bu_j)(\bu_i \wedge \bu_j)^* \big] }= d_\bE(\x,\y).
\end{align*}

A particular example of such an operator $\bE$ is as follows:
Assume that $O\in\rS_+(\cH)$ with the eigenvalues $\lambda_i$ and eigenvectors $\bu_i$,
\begin{equation}\label{eigvalvecid}
\begin{aligned}
O\bu_i=\lambda_i\bu_i, \quad \lambda_1\ge\ldots\ge\lambda_n\ge 0, \quad \langle \bu_i,\bu_j\rangle=\delta_{ij}, \text{ for } i,j\in[n].
\end{aligned}
\end{equation}
Then, $O$ induces an operator $\bE=O\wedge O$ such that $(O\wedge O)(\x\wedge\y)=(O\x)\wedge(O\y)$ (see \cite{Frib}).   Hence,
\begin{equation*}
(O\wedge O)(\bu_i\wedge \bu_j)=\lambda_i\lambda_j \bu_i\wedge\bu_j, \quad 1\le i<j\le n.
\end{equation*}
That is,  $d_{ij}(O\wedge O)=\lambda_{i}\lambda_j$ for $i< j$. We shall call such operators $\bE$ \emph{reducible}.


\medskip

Let $X=\{\bp_1,\ldots,\bp_n\} \subseteq \R^N$ be a set of $n$ points equipped with a distance function $d$. (Observe that it is sufficient to consider $N\leq n-1$.) Then, the values of $d$ can be organised into a symmetric matrix with zero diagonal, $D(\bE)=[d_{ij}] \vc [d(\bp_i,\bp_j)] \in\R^{n\times n}_+$.  It leads to the following concept:

\begin{definition}\label{defdsitmat} A real $n\times n$ matrix $D= [d_{ij}]\in\R^{n\times n}$ is called  a \emph{distance matrix} if it satisfies the following properties:
\begin{enumerate}[(a)]
\item $D$ is symmetric: $D^\top=D$.
\item $D$ is positive semi-definite: $d_{ij}\ge 0$, for $i,j\in[n]$.
\item $D$ has zero diagonal: $d_{ii}=0$ for $i\in[n]$.
\item The triangle inequality holds:
\begin{equation}\label{triadmat}
d_{ij}+d_{jk}\ge d_{ik}, \quad  \textrm{ for all } i,j,k\in[n].
\end{equation}
\end{enumerate}
We denote by $\cD_n\subset \R^{n\times n}$ the set of all $n\times n$ distance matrices. A distance matrix is called positive if all its off-diagonal entries are positive.
A distance matrix $D \in \cD_n$ is \emph{$\ell_p^N$-induced} 
if  $d_{ij}=\|\bp_i-\bp_j\|_p$ for some $\bp_1,\ldots,\bp_n\in\R^N$ and $p\in[1,\infty]$. 
\end{definition}

For a symmetric nonnegative matrix $S\cv S^{\circ 1}$ denote
\begin{equation}\label{defScircp}
S^{\circ p}\vc[s_{ij}^p], \quad s_{ij}^p=s_{ji}^p\ge 0, \quad \text{ for } i,j\in[n] \text{ and } p>0.
\end{equation}
It is straightforward to show
\begin{equation}\label{pdmat}
D\in\cD_n\Rightarrow D^{\circ p}\in \cD_n \textrm{ for } p\in(0,1).
\end{equation}

\begin{definition}\label{def:ind}
We say that an operator $\bE \in \rS_+(\Lambda^2(\cH))$ is \emph{induced by a distance function} $d$ on the $n$-point set $X$ if $\bE$ is of the form \eqref{Euiujeig} and $D(\bE)$ is a distance matrix.
\end{definition}
Conjecture \ref{conj} asserts that every triangular operator $\bE \in \cT(\cH)$ of the form \eqref{Euiujeig} on a Hilbert space of dimension $n$ is induced by some distance function on an $n$-point set $X$.

Note that $D(\bE^p)=D^{\circ p}(\bE)$, which means that if $\bE \in \cT(\cH)$ is of the form \eqref{Euiujeig} then $\bE^p \in \cT(\cH)$ for all $p\in(0,1]$. Recall also that Proposition \ref{prop:sym} implies that if an operator $\bE$ is triangular then so is $\bE^U \vc (U\wedge U)^*\, \bE\, (U \wedge U)$, for any `local' unitary $U \in \rB(\cH)$. Clearly, if $\bE$ is of the form \eqref{Euiujeig} then so is $\bE^U$. 

Observe that any operator $\widetilde{\bE} \in \rS_+(\Lambda^2(\cH))$ can be written as a `global' rotation of an operator $\bE$ of the form \eqref{Euiujeig}. That is, for any $\widetilde{\bE} \in \rS_+(\Lambda^2(\cH))$ there exist a unitary $V \in \rB(\Lambda^2(\cH))$ and an operator $\bE \in \rS_+(\Lambda^2(\cH))$ of the form \eqref{Euiujeig} such that $\widetilde{\bE} = V^* \bE V$. However, such a `global' rotation of a triangular operator $\bE$ will not, in general, yield a triangular operator, regardless of whether $\bE$ was of the form \eqref{Euiujeig} or not.

It is easy to provide a counterexample: Let $\dim \cH = 4$ and let $\bE$ be induced by the following distance matrix:
\begin{equation}
\label{eig_ex1}
\begin{aligned}
&d_{12} = 2,   &&  d_{13} = 3, && d_{23} = 1, \\
&d_{14} = 1, && d_{24} = 3, &&   d_{34} = 2.
\end{aligned}
\end{equation}
Then, $\bE$ is triangular. Let now $V \in \rB(\Lambda^2(\cH))$ be defined as
\begin{align*}
    V (\bu_{i} \wedge \bu_j ) = \begin{cases} 
    \bu_1 \wedge \bu_4, \; \text{ if } i = 1, j = 3, \\ 
    \bu_1 \wedge \bu_3, \; \text{ if } i = 1, j = 4, \\
    \bu_{i} \wedge \bu_j, \; \text{ otherwise}.
    \end{cases}
\end{align*}
Then, $\widetilde{\bE} = V^* \bE V$ is \emph{not} triangular, because 
\begin{align*}
    d_{\widetilde{\bE}}(\bu_1,\bu_2) + d_{\widetilde{\bE}}(\bu_2,\bu_3) = 3 > d_{\widetilde{\bE}}(\bu_1,\bu_3) = 1. 
\end{align*}

In the next subsections we provide analytical evidence for the general validity of Conjecture \ref{conj}. For further, numerical, evidence see Appendix~\ref{sec:numerics}.

\subsection{Triangular operators induced by the Euclidean distance}

Let us start with triangular operators induced by the `line geometry', that is $\ell_2^1$.

\begin{theorem}\label{Bistron}  Let $\bE\in\rS_+(\Lambda(\cH))$ be induced by the $\ell_2^1$ distance, that is 
\begin{equation}\label{defEBis}
d_{ij}=|\bp_i-\bp_j| > 0, \quad \text{ for }\; i,j\in [n],
\end{equation}
with $\bp_1,\ldots,\bp_n\in\R$. Then, $\bE$ is a triangular operator for any $n \geq 2$.
\end{theorem}
\begin{proof} 
For $n = 2$ the assertion follows from Proposition \ref{proplindep}, while for $n=3$ it follows from Theorem \ref{minxyzthm}. Let us then assume that $n>3$ and, on the strength of Corollary \ref{cornecsufcondti}, let $\x,\y,\z \in \mathbb{P}(\C^n)$ be three orthonormal vectors.

Let us set $\bH=\sum_{i=1}^n \bp_i\bu_i\bu_i^*\in\rS(\C^n)$ and compute
\begin{equation*}
\begin{aligned}
d_\bE(\x,\y)^2 & = \frac{1}{2}\sum_{i,j=1} ^n(\bp_i - \bp_j)^2 |x_i y_j - x_j y_i|^2= 
\sum_{i,j=1}^n (\bp_i^2 - \bp_i \bp_j) |x_i y_j - x_j y_i|^2\\
& =\sum_{i,j=1}^n (\bp_i^2- \bp_i \bp_j) \left(|x_i|^2 |y_j|^2 - x_iy_j\bar{x}_j\bar{y}_i - x_jy_i\bar{x}_i\bar{y}_j + |y_i|^2 |x_j|^2\right) \\
& =\langle \bH^2\x,\x\rangle+\langle \bH^2\y,\y\rangle 
-2 \big(\langle \bH\x,\x\rangle\langle \bH\y,\y\rangle - |\langle \bH\x,\y\rangle|^2 \big).
\end{aligned}
\end{equation*}
As $\x,\y,\z \in \cH$ are three orthonomal vectors we can complete them to an orthonormal basis $\x=\bv_1,\y=\bv_2,\z=\bv_3,\ldots,\bv_n$.  
In this basis $\bH$ and $\bH^2$ are respectively represented by positive definite Hermitian matrices $H=[h_{ij}], H^2=[h_{ij}^{(2)}]\in\C^{n\times n}$.
Hence,
\begin{equation*}
\begin{aligned}
d(\x,\y)=\sqrt{h_{11}^{(2)}+h_{22}^{(2)} -2(h_{11}h_{22}-|h_{12}|^2)},\\
d(\x,\z)=\sqrt{h_{11}^{(2)}+h_{33}^{(2)} -2(h_{11}h_{33}-|h_{13}|^2)},\\
d(\y,\z)=\sqrt{h_{22}^{(2)}+h_{33}^{(2)} -2(h_{22}h_{33}-|h_{23}|^2)}.
\end{aligned}
\end{equation*}
We thus have
\begin{align*}
  d_\bE(\x,\z)+d_\bE(\z,\y)- d_\bE(\x,\y) & =  -2 \big(h_{11}h_{33}-|h_{13}|^2+h_{22}h_{33}-|h_{13}|^2 \big)+\\
& \quad +2h_{33}^{(2)}+2 \big(h_{11}h_{22}-|h_{12}|^2 \big) + \\
& \quad + 2\sqrt{h_{11}^{(2)}+h_{33}^{(2)} -2 \big(h_{11}h_{33}-|h_{13}|^2 \big)} \times \\
& \qquad \times \sqrt{h_{22}^{(2)}+h_{33}^{(2)} -2 \big(h_{22}h_{33}-|h_{23}|^2 \big)}.
\end{align*}
Now, let us define the matrix $G \in \C^{3\times 3}$, as $G=[g_{ij}]=[h_{ij}]$, for $i,j\in[3]$, then  $G^2=[g_{ij}^{(2)}]$.   Clearly, $g^{(2)}_{ii}\le h_{ii}^{(2)}$ for $i\in[3]$ and $g_{ij} = h_{ij}$, thus
\begin{align*}
  d_\bE(\x,\z)+d_\bE(\z,\y)- d_\bE(\x,\y) & \geq  -2 \big(g_{11}g_{33}-|g_{13}|^2+g_{22}g_{33}-|g_{13}|^2 \big)+\\
& \quad +2g_{33}^{(2)}+2 \big(g_{11}g_{22}-|g_{12}|^2 \big) + \\
& \quad + 2\sqrt{g_{11}^{(2)}+g_{33}^{(2)} -2 \big(g_{11}g_{33}-|g_{13}|^2 \big)} \times \\
& \qquad \times \sqrt{g_{22}^{(2)}+g_{33}^{(2)} -2 \big(g_{22}g_{33}-|g_{23}|^2 \big)} \\
& = d_\bG(\x,\z)+d_\bG(\z,\y)- d_\bG(\x,\y) \geq 0.
\end{align*}
where $\bG=\sum_{i=1}^3 G_i\bw_i$ for some orthonormal basis $\bw_1,\bw_2,\bw_3$ in $\C^3$. But the operator $\bG$ is triangular, hence the last inequality follows.
\end{proof}

Thanks to the 2-cone structure of the set of triangular operators, the above theorem can be easily generalised to operators induced by the general Euclidean distance.

\begin{theorem}\label{subscH}  Let $\bE\in\rS_+(\Lambda(\cH))$ be induced by the $\ell_2^N$ distance, that is 
\begin{equation}\label{defEBis}
d_{ij}=\|\bp_i-\bp_j\|_2 > 0, \quad \text{ for }\; i,j\in [n],
\end{equation}
with $\bp_1,\ldots,\bp_n\in\R^N$. Then, $\bE$ is a triangular operator for any $n \geq 2$ and any $N \leq n-1$.
\end{theorem}
\begin{proof}
We have
\begin{equation*}
d_{ij}=\sqrt{\sum_{k=1}^{N}\big(\bp_{i}^k-\bp_{j}^k\big)^2}, \quad \text{ with } \; \bp_{i}=(\bp_{i}^1,\ldots,\bp_{i}^N)^\top\in\R^{N}, \text{ for } i\in[n].
\end{equation*}
Set
\begin{equation*}
\bE_k=\sum_{1\le i< j\le n} \big|\bp_{i}^k-\bp_{j}^k \big| \, \bu_i\wedge\bu_j, \quad \text{ for } \; k\in[N].
\end{equation*}
Theorem \ref{Bistron} yields that $\bE_k\in\cT(\cH)$, for all $k\in[N]$. Then, Proposition \ref{cT2cone} implies that $\bE=\sqrt{\sum_{k=1}^N \bE_k^2} \in \cT(\cH)$.  

Observe that by translation $\bp\mapsto \bp+\ba$ we can assume that $\bp_1=\0$.
As span$(\bp_2,\ldots,\bp_n)$ is at most of dimension $n-1$, we can assume without loss of generality that $N \leq n-1$. 
\end{proof}


We now recall the known results due to Schoenberg \cite[Theorem 1]{Sch35}, see also \cite{Mae13}, that gives simple necessary and sufficient conditions on a distance on $[n]$ points, which is realised 
as the Euclidean distance on some $n$ distinct points in $\x_1,\ldots,\x_n\in\R^N$.
\begin{theorem}\label{Sch35}  Let $D=[d_{ij}]\in\cD_n$.  Then there exists $\x_1,\ldots,\x_n\in\R^{n-1}$  such that 
$d_{ij}=\|\x_i-\x_j\|$ for $i,j\in [n]$ if and only if the following symmetric matrix $A=[a_{ij}]\in \R^{(n-1)\times (n-1)}$ is positive semidefinite:
\begin{equation}\label{Schmat}
a_{ij}=\frac{1}{2}\big(d^2_{1(i+1)}+d^2_{1(j+1)}-d^2_{(i+1)(j+1)}\big), \quad i,j\in[n-1].
\end{equation}
Assume that $A$ is positive semidefinite with $\rank A=r\in[n-1]$.
Then $\x_1,\ldots,\x_n\in\R^r$ but not in $\R^{r-1}$.
\end{theorem}

Assume that $d(i,j), i,j\in[n]$ is a distance on $[n]$.  It is straightforward to show that for each $\gamma\in(0,1)$ the function $d(i,j)^\gamma$ is a distance on $[n]$.  (Concavity of the function $x^\gamma+y^\gamma$.)
In the paper \cite[Theorem 3]{Sch37} it is shown that if $d(i,j), i,j\in[n]$ is a distance on $[n]$ induced by the Euclidean  distance on $
\R^{n-1},$ then the distance $d^\gamma$ is induced also by the Euclidean distance on $\R^{n-1}$.

The following theorem is proven by Schoenberg \cite[Theorem 5]{Sch38}.
\begin{theorem}\label{ellppowgam} Let $\y_1,\ldots,\y_n\in \R^N$.  Assume that $p\in[1,2]$, and $D=[\|\y_i-\y_j\|_p]$.  Then, for $\gamma\in[0,p/2]$ the distance  matrix $D^{\circ \gamma} $  induced by $\ell_2^{n-1}$.
\end{theorem} 
Together with Theorem \ref{subscH} it yields the following result:
\begin{corollary}\label{gamcor(a)}  Let the operator $\bE\in \rS_+(\Lambda_2(\cH))$ be induced by $\ell_p^N$ for $p\in[1,2]$. Then $\bE^{\gamma}\in \cT(\cH)$, for all $\gamma\in(0,p/2]$. 
\end{corollary}
{

%

\subsection{Reducible triangular operators}\label{subscH2}
Let us now turn to the case of reducible triangular operators, that is $\bE = O \wedge O$ for some $O \in \rB(\cH)$.
Denote 
\begin{equation*}
\begin{aligned}
\Delta(\x)=\x\x^\top-\diag(x_1^2,\ldots,x_n^2), \quad \x=(x_1,\ldots,x_n)^\top\in\R^n.
\end{aligned}
\end{equation*}
\begin{theorem}\label{thm:reducible}  Assume that $\dim\cH\ge 3$.  Let $\bE\in \rS_+(\Lambda_2(\cH))$  and $D(\bE)$ is a distance matrix.  If $\bE$ is reducible,
then $\bE$ is a triangular operator.
\end{theorem}
To prove this theorem we will need the following lemma:
\begin{lemma}\label{ildlemma}
Assume that $n\ge 4$ and
\begin{equation}\label{abineq}
\begin{aligned}
&\ba=(a_1,\ldots,a_n)^\top\in\R^n, \quad \bb=(b_1,\ldots,b_{n-2})^\top \in \R^{n-1},\\
&a_i\ge b_i\ge a_{i+1}\ge 0 \textrm{ for } i\in[n-1].
\end{aligned}
\end{equation}
Assume that $\Delta(\ba)\in\cD_n$.  Then $\Delta(\bb)\in\cD_{n-1}$.
\end{lemma}
\begin{proof}  Assume that $D=[d_{ij}]\in\cD_n$.  Suppose that $d_{ij}=0$ for $i\ne j$.
It is straightforward to show that $d_{ik}=d_{jk}$ for all $k\in[n]$. 
Assume that  $\Delta(\ba)\in\cD_n$, and the inequalities \eqref{abineq} hold. 

Suppose that 
$a_n=0$.  Then $a_ia_n=0$ for $i\in[n-1]$.  Hence, $\Delta(\ba)=0$. Therefore, $a_2=\ldots=a_n=0$, and $b_2=\ldots=b_{n-1}=0$.  Thus $\Delta(\bb)=0\in\cD_{n-1}$.

We now assume that $a_n>0$ and prove the lemma by induction.
Observe that for $1\le i<j<k\le n$ we have the equality
$a_ia_j\ge \max(a_ia_j, a_ia_k, a_ja_k)$.
Assume that $\ba$ in \eqref{abineq} is fixed.
 Let $n=4$.   Observe:
 \begin{equation*}
 \begin{aligned}
b_1b_3+b_2b_3-b_1b_2\ge b_1a_4+b_2a_4-b_1b_2=b_1(a_4-b_2)+b_2a_4\ge\\
a_1(a_4-b_2)+b_2a_4=b_2(a_4-a_1)+a_1a_4\ge a_2(a_4-a_1)+a_1a_4\ge 0.
\end{aligned}
\end{equation*}
The last inequality follows from the assumption that $\Delta(\ba)\in\cD_4$.

Assume now that lemma holds for $n=m\ge 3$ and assume that $n=m+1$.  We claim that the following inequality holds
\begin{equation*}
b_ib_k+b_kb_j\ge b_ib_j \textrm{ for } 1\le i<j\le k\le n.
\end{equation*}
Suppose first that $k\le n-1$.  Then this inequality follows from  induction hypothesis
for $\ba'=(a_1,\ldots,a_{n-1})^\top, \bb'=(b_1,\ldots,b_{n-2})^\top$.
Similarly, we deduce the above inequality if $i\ge 2$.  Thus, we are left with the case $i=1<j<k=n$.   This inequality follows from teh same arguments we used for the case $n=4$.
\end{proof}

{\it Proof of Theorem \ref{thm:reducible}.}
Assume that $\cH=\C^n$,  where $n\ge 3$.
Suppose that $O$ satisfies \eqref{eigvalvecid} and $\bE=O\wedge O$.  
 We prove that the theorem in that case by induction on $n$.  For $n=3$ our theorem follows from Theorem \ref{minxyzthm}.  Assume that our theorem holds for $n=m\ge 3$, and let $n=m+1$.  Corollary \ref{cornecsufcondti} states that it is enough to show 
 the triangle inequality every triple of orthonormal vectors $\x,\y,\z\in\C^n$.
 Let $\W$ be an $n-1$-dimensional subspace of $\C^n$, such that $\x,\y,\x\in\W$.
 Denote  by $O(\W):\W\to\W$ a positive semidefinite operator so that the restriction of the hermitian form $O^2$ to  $\W$ is equal to the hermitian form of $O(\W)^2$.   Then the hermitian form given  $O^2(\W)\wedge O^2(\W)$ is the restriction of the hermitian form $O^2\wedge O^2$ to $\W\wedge\W$.   

Assume that $\mu_1\ge\cdots\ge\mu_{n-1}$ are the eigenvalues of $O(\W)$.  Recall the interlacing inequalities
for $O^2(\W)$ and $O^2$, where we denote the eigenvalues of $O^2(\W)$ by $\mu_1^2\ge\cdots\ge \mu_{n-1}^2$ :
\begin{equation*}
\lambda_1^2\ge \mu_1^2\ge \lambda_2^2\ge\cdots\ge \lambda_{n-1}^2\ge \mu_{n-1}^2\ge \lambda_n ^2\ge 0.
\end{equation*}
Hence
\begin{equation*}
\lambda_1\ge \mu_1\ge \lambda_2\ge\cdots\ge \lambda_{n-1}\ge \mu_{n-1}\ge \lambda_n \ge 0.
\end{equation*}
Since $D(\bE)$ is a distance matrix, Lemma \ref{ildlemma} yields that $\Delta((\mu_1,\cdots,\mu_{n-1})^\top)\in \cD_{n-1}$.  The induction hypothesis yields that $O(\W)\wedge O(\W)\in\cT(\W)$.
Hence the tirangle inequality holds for $\x,\y,\z$.
\qed 

Below we present an example which shows that Theorem \ref{thm:reducible} does not follow from Theorem  \ref{subscH}}  for $n=4$.
Let 
\begin{equation}\label{ccond}
b_1=b_2=1, \quad b_3 =1/2, \quad 2 \ge b_4\ge 1.
\end{equation}
Set  $\bb=(b_1,b_2,b_3,b_4)^\top$ and $D=[d_{ij}]=\Delta(\bb)$.  It is straightforward to check that $D\in \cD_4$.
Assume that the eigenvalues of $O\in\rS_+(\C^4)$ are $b_i,i\in[4]$.
Then the matrix $A$ in Theorem \ref{Sch35} is
\begin{equation*}
\begin{bmatrix}1&\frac{1}{2}&\frac{1}{2}\\ \frac{1}{2}&\frac{1}{4}&\frac{1+3b_4^2}{8}\\
\frac{1}{2}&\frac{1+3b_4^2}{8}&b^2\end{bmatrix}.
\end{equation*}
It is straightforward to show that the conditions \eqref{ccond} yield that $\det A<0$.
For these parameters, the distance matrix $D$  can't be realized as Euclidean distance on $\R^3$.  
\subsection{Other examples of triangular operators}\label{subscH3}

\begin{example}\label{squarell1}
Example 
Consider a square with  $4$-vertices in the plane:
\begin{equation*}
\x_1=(0,0)^\top, \x_2=(1,0)^\top, \x_3=(0,1)^\top, \x_4=(1,1)^\top
\end{equation*}
Let $D=[d_{ij}]$ be the distance matrix induced by $\|\x_i-\x_j\|_1$:
\begin{equation*}
\begin{bmatrix}0&1&1&2\\1&0&2&1\\1&2&0&1\\2&1&1&0\end{bmatrix}.
\end{equation*}
Then $\bE$ constructed form $D$ via \eqref{d(x,y)form} satisfies the triangle inequality \eqref{triangins}.
\end{example}
\begin{proof}  Observe that the eigenvalues of $\E$ are $(2,2,1,1,1,1)$, thus the sufficient condition from Proposition \ref{sufcond} is satisfied.
\end{proof}

Finally, let us observe that the set of triangular operators is not limited to the operators of the form \eqref{Euiujeig}. Indeed, take $\bE_1, \bE_2 \in \cT(\cH)$, which are both of the form \eqref{Euiujeig}, but do not commute, i.e.
\begin{align*}
    \bE_1(\bu_i\wedge\bu_j)=d_{ij} \, \bu_i\wedge\bu_j, && \bE_1(\bv_i\wedge\bv_j)=d'_{ij} \, \bv_i\wedge\bv_j,
\end{align*}
for two different orthonormal bases $\{\bu_i\}_i, \{\bv_j\}_j$ of $\cH$. Then, by Proposition~\ref{cT2cone}, the operator $\bE=\sqrt{\bE_1^2 + \bE_2^2}$ is triangular, but it is not necessarily of the form \eqref{Euiujeig}.

%




\emph{Acknowledgements:}
Financial support by NCN under the Quantera project no. 2021/03/Y/ST2/00193,
the Foundation for Polish Science 
under the Team-Net project no. POIR.04.04.00-00-17C1/18-00,
and by Simons collaboration grant for mathematicians
is gratefully acknowledged.

\emph{Declarations of interest:} none
\


\appendix

\section{\label{sec:lemma}A trigonometrical lemma}

In this Appendix we prove a technical lemma, which is used in the proof of Theorem \ref{minxyzthm}.

\begin{lemma}\label{ineqforti} Let $a,b,t>0,$ an consider the following function $F$ and its minimum:
\begin{equation}\label{ineqfort1i} 
\begin{aligned}
F(a,b,t,\theta,\phi) & = \sqrt{a^2\cos^2\theta+b^2\sin^2\theta}\,+\sqrt{a^2\cos^2\phi+b^2\sin^2\phi} \; - \\
& \hspace*{1.9cm} -(a+b)t\sin(\theta+\phi), \quad  \text{ for } \;  \theta,\phi\in[0,\pi/2],\\
\omega(a,b,t) & =\min_{\theta,\phi\in[0,\pi/2]}F(a,b,t,\theta,\phi).
\end{aligned}
\end{equation}
If $a=b$ then 
\begin{equation}\label{omega=b}
\omega(a,a,t)=2a(1-t)
\end{equation}
and the minimum is attained 
if and only if $\theta+\phi=\pi/2$.
\medskip

\noindent If $a\ne b$ then the following conditions hold:
\begin{enumerate}[(a)]
\item For $t\in(0,1)$ the inequality $\omega(a,b)>0$ holds.
\item For $t=1$ the equality $\omega(a,b)=0$ holds. It is achieved if and only if 
\begin{equation}\label{F=0cond}
\begin{aligned}
& \text{either }  \theta+\phi=\pi/2 
\text{ or }  \tan \theta\tan\phi=\frac{a}{b}, \quad \text{with } \theta,\phi\in(0,\pi/2).
\end{aligned}
\end{equation}

\item For $t>1$ the equality $\omega(a,b)=(a+b)(1-t)$ holds.  
It is achieved if and only if $\theta+\phi=\pi/2$, with $\theta,\phi\in(0,\pi/2)$.
\end{enumerate}
\end{lemma}
\begin{proof}  Clearly,
\begin{equation*}
F(a,a,t,\theta,\phi)=2a-2at\sin(\theta+\phi)\ge 2a-2at.
\end{equation*}
As $(\phi,\theta)\in[0,\pi/2]$, equality holds if and only if $\theta+\phi=\pi/2$.

Assume that $a\ne b$. Clearly,
\begin{equation}\label{omineqt>1}
\omega(a,b,t)\le (a+b)(1-t)=F(a,b,t,0,\pi/2)=F(a,b,t,\pi/2,0).
\end{equation}

We first study the case $0<t\le 1$.
Let us consider the critical points of $F(a,b,t,\theta,\phi)$ in $(0,\pi/2)^2$:
\begin{equation*}
\begin{aligned}
& \frac{(b-a)\sin\theta\cos\theta}{\sqrt{a^2\cos^2\theta+b^2\sin^2\theta}}=t\cos(\theta+\phi),\\
& \frac{(b-a)\sin\phi\cos\phi}{\sqrt{a^2\cos^2\phi+b^2\sin^2\phi}}=t\cos(\theta+\phi),
\end{aligned}
\end{equation*}
Hence, the two left hand sides of the above equality are equal.  Subtract the square of the first left hand side from the square of the second left hand side, divide by $(b-a)^2$,  and take the common denominator to deduce:
\begin{align*}
0 & =\sin^2\varphi \cos^2\varphi\big(a^2\cos^2\theta + b^2 \sin^2\theta\big)-\sin^2\theta \cos^2\theta\big(a^2\cos^2\varphi + b^2 \sin^2\varphi\big)\\
& = a^2\cos^2\varphi\cos^2\theta(\sin^2\varphi-\sin^2\theta) +b^2\sin^2\varphi\sin^2\theta( \cos^2\varphi-\cos^2\theta)\\
& = a^2\cos^2\varphi\cos^2\theta(\sin^2\varphi-\sin^2\theta) +b^2\sin^2\varphi\sin^2\theta( -\sin^2\varphi+\sin^2\theta)\\
& = a^2\cos^2\varphi\cos^2\theta(\sin^2\varphi-\sin^2\theta) -b^2\sin^2\varphi\sin^2\theta( \sin^2\varphi-\sin^2\theta)\\
& = (\sin^2 \theta - \sin^2 \varphi)(b^2 \sin^2\theta \sin^2 \varphi - a^2 \cos^2\theta\cos^2\varphi).
\end{align*}
to deduce that one of the equalities hold:
\begin{align*}
\sin^2\theta-\sin^2\phi=(\sin \theta-\sin\phi)(\sin \theta+\sin\phi)=0,\\
b^2 \sin^2\theta \sin^2 \varphi - a^2 \cos^2\theta\cos^2\varphi= 0.
\end{align*}
Since we assumed that $\theta,\phi\in(0,\pi/2)$ we have the following two possiblities
\begin{equation}\label{check5}
\begin{aligned}
\text{either } \theta=\phi\in (0,\pi/2) 
\quad \text{or} \quad \tan \theta\tan\phi=\frac{a}{b}, \text{ with } \theta,\phi\in (0,\pi/2).
\end{aligned}
\end{equation}
Let us first consider the second possibility in \eqref{check5}:
\begin{equation*}
\begin{aligned}
& \tan \phi =\frac{a}{b}\cot\theta \; \Rightarrow \; \sqrt{a^2\cos^2\theta+b^2\sin^2\theta}=b \sin\theta \sqrt{\tan^2\phi+1}=\frac{b\sin\theta}{\cos\phi},\\
& \tan\theta=\frac{a}{b}\cot\phi \; \Rightarrow \;  \sqrt{a^2\cos^2\phi+b^2\sin^2\phi}=\frac{b\sin\phi}{\cos\theta}
\end{aligned}
\end{equation*}
Use the equality $a=b\tan\phi\tan\theta$ to deduce
\begin{equation*}
\begin{aligned}
\cos\theta\cos\phi F(a,b,\theta,\phi) & = b\, (\sin\theta\cos\theta+\sin\phi\cos\phi)-\\
& \quad - b\, (\cos\theta\cos\phi+\sin\theta\sin\phi) \,t\, (\sin\theta\cos\phi+\sin\phi\cos\theta)\\
& = b\, (\sin\theta\cos\theta+\sin\phi\cos\theta)(1-t).
\end{aligned}
\end{equation*}
Hence,
\begin{equation}\label{Fineqsec}
\begin{aligned}
& F(a,b,t,\theta,\phi)\begin{cases}
> 0, \quad \textrm{ if } t\in(0,1),\\
=0, \quad \textrm{ if } t=1,
\end{cases}\\
& \textrm{for} \quad \tan\theta\tan\phi=\frac{a}{b} \textrm{ and }\theta,\phi\in (0,\frac{\pi}{2}).
\end{aligned}
\end{equation}
We now discuss the first possibility in \eqref{check5}:
\begin{equation*}
F(a,b,t,\theta,\theta)=2\sqrt{a^2\cos^2\theta+b^2\sin^2\theta}\,-(a+b)t \sin(2\theta), \text{ with }\theta\in (0,\pi/2).
\end{equation*}
Without loss of generality let us assume that $a+b=1$. Fix $\theta\in (0,\pi/2)$,
and consider the function $F(a,b,t,\theta,\theta)$ for $a,b\ge 0$ and $a+b=1$.  Recall that this function is strictly convex on the above interval in $\R^2$.   Let us use Lagrange
multipliers to find the critical point of $F(a,b,\theta,\theta)$:
\begin{equation*}
F_a(a,b,\theta,\theta)=\frac{a\cos^2\theta}{F(a,b,\theta,\theta)}=F_b(a,b,\theta,\theta)=\frac{b\sin^2\theta}{F(a,b,\theta,\theta)} \; \Rightarrow \; \tan^2\theta=\frac{a}{b}.
\end{equation*}
This case is a special case of the second possibility  in \eqref{check5},  hence \eqref{Fineqsec} applies.

We now  consider the boundary cases when at least $\theta$ or $\phi$ are in the set $\{0,\frac{\pi}{2}\}$.  Let us assume that $\theta=0$.  Then,
\begin{equation*}
\begin{aligned}
F(a,b,t,0,\phi) & = a+\sqrt{a^2\cos^2\phi+b^2\sin^2\phi}\, -(a+b)t\sin\phi\\ 
& \ge a+b\sin\phi-(a+b)t\sin\phi \\
& \ge a(1-t\sin\phi)+b(1-t)\sin\phi\ge a(1-t)
\end{aligned}
\end{equation*}
Hence, for $0<t<1$, we have $F(a,b,t,0,\phi)>0$ for all $\phi\in[0,\pi]$.
For $t=1$ we have the inequality $F(a,b,t,0,\phi)\ge 0$ for all $\phi\in[0,\pi]$.  Equality is achieved for $\phi=\frac{\pi}{2}$.  Similar results apply to other boundary cases.
These results establish the cases (a) and (b).

Assume that $t>1$.  Then,
\begin{equation*}
F(a,b,t,\theta,\phi)=F(a,b,1,\theta,\phi)-(a+b)(t-1)\sin(\theta+\phi)\ge (a+b)(1-t).
\end{equation*}
Hence, equality holds in \eqref{omineqt>1}.  Furthermore,  $\omega(a,b,t)=F(a,b,t,\theta,\phi)$ if and only if $\theta+\phi=\pi/2$ and one of the equalities in \eqref{check5} holds.
We now show that the second equality in  \eqref{check5} does not hold.
Indeed, as $\phi=\pi/2-\theta$ we get that $\tan\phi=\cot\theta$ and $\tan\theta\tan\phi=1\ne b/a$.   This proves the case (c).
\end{proof}

\section{\label{sec:numerics}Numerical evidence in favour of Conjecture \ref{conj}}

In this Appendix, the results of numerical calculations are presented. We performed various Monte Carlo simulations on a few stages of our work to look for possible counterexamples to proposed statements. Below we present the general results in favour of Conjecture \ref{conj}.
We focus on $4$ different tests: first the general check of the triangle inequality for arbitrary distances (test 1) and for distances generated by the $L_{\infty}$ norm (test 2). Then, using Corollary \ref{cornecsufcondti} we repeat those checks, but restricting ourselves to orthonormal vectors (arbitrary distances -- test 3, and $L_{\infty}$-induces distances -- test 4).

First, we drew three random vectors with Haar measure and orthogonalized them by the Gram-Schmidt procedure (if necessary). 

Then we constructed a distance matrix by drawing random values for consecutive entries from the uniform distribution on $[0,1]$, saving only those that satisfying the triangle inequalities. For the $L_{\infty}$ we drew $n$ random points from the uniform distribution on $[0,1]^{2 n}$ and then calculated $L_{\infty}$ distances between them. Next we substituted the obtained formulas into triangle inequality \eqref{triangins} with the distances calculated from \eqref{d(x,y)form} and looked for the minimal difference between the left-hand side and the right-hand side, which is equivalent to looking for the solution of minimum problem \eqref{minprob}.

\begin{table}[h]
\caption{\label{tab:num1}{\small The results of the Monte Carlo simulations discussed in the text. ``Samples'' refers to the number of drawn quadruples ($\x,\y,\z, D$). ``Min for random $D$'' to the minimal value obtained for random distance matrices $E$ and ``Min for $D$ from $L_{\infty}$'' to minimal value obtained for random distance matrices generated from distances between points in $L_{\infty}$ space. }}
{\small
\begin{tabular}{l|lll|lll}
\multirow{2}{*}{Dim.} & \multicolumn{3}{l|}{Random vectors}& \multicolumn{3}{l}{Random orthonormal vectors} \\
\cline{2-7} 
& \multicolumn{1}{l|}{Samples} & \multicolumn{1}{l|}{\begin{tabular}[c]{@{}l@{}}Min for \\  random $D$\end{tabular}} & \begin{tabular}[c]{@{}l@{}}Min for \\
$D$ from $L_{\infty}$\end{tabular} & \multicolumn{1}{l|}{Samples} & \multicolumn{1}{l|}{\begin{tabular}[c]{@{}l@{}}Min for \\
random $E$  \end{tabular}} & \begin{tabular}[c]{@{}l@{}}Min for \\
$E$ from $L_{\infty}$\end{tabular} \\ \cline{1-7} 
3  & \multicolumn{1}{l|}{$2.4 \cdot 10^{7}$}& \multicolumn{1}{l|}{$4.575 \cdot 10^{-5}$}& $9.561 \cdot 10^{-3}$   & \multicolumn{1}{l|}{$2 \cdot 10^{7}$}    & \multicolumn{1}{l|}{$2.212\cdot 10^{-3}$} & $2.190 \cdot 10^{-2}$ \\
4 & \multicolumn{1}{l|}{$2.4 \cdot 10^{7}$} & \multicolumn{1}{l|}{$4.011 \cdot 10^{-3}$} & $7.706\cdot 10^{-2}$   & \multicolumn{1}{l|}{$2 \cdot 10^{7}$} & \multicolumn{1}{l|}{$1.327\cdot 10^{-2}$} & $3.949 \cdot10^{-1}$ \\
5 & \multicolumn{1}{l|}{$2.4 \cdot 10^{7}$} & \multicolumn{1}{l|}{$9.938 \cdot 10^{-3}$} & $3.602 \cdot 10^{-1}$  & \multicolumn{1}{l|}{$2 \cdot 10^{7}$} & \multicolumn{1}{l|}{$4.079 \cdot 10^{-2}$} & $8.975\cdot10^{-1}$ \\
6 & \multicolumn{1}{l|}{$2.4 \cdot 10^{7}$} & \multicolumn{1}{l|}{$4.329 \cdot 10^{-2}$} & $5.972 \cdot 10^{-1}$  & \multicolumn{1}{l|}{$2 \cdot 10^{7}$} & \multicolumn{1}{l|}{$7.337 \cdot 10^{-2}$} & $1.599$  \\
7 & \multicolumn{1}{l|}{$2.4 \cdot 10^{7}$} & \multicolumn{1}{l|}{$6.163 \cdot 10^{-2}$} & $1.176$ & \multicolumn{1}{l|}{$2 \cdot 10^{7}$} & \multicolumn{1}{l|}{$1.120\cdot 10^{-1}$} & $3.144$ \\
8 & \multicolumn{1}{l|}{$2.4 \cdot 10^{7}$} & \multicolumn{1}{l|}{$8.165 \cdot 10^{-2}$} & $1.645$ & \multicolumn{1}{l|}{$2 \cdot 10^{7}$} & \multicolumn{1}{l|}{$1.324\cdot 10^{-1}$} & $3.948$ \\
9 & \multicolumn{1}{l|}{$7.2 \cdot 10^{6}$}      & \multicolumn{1}{l|}{$7.865 \cdot 10^{-2}$} & $2.615$ & \multicolumn{1}{l|}{$6 \cdot 10^{6}$} & \multicolumn{1}{l|}{$1.796 \cdot 10^{-1}$} & $4.580$ \\
10 & \multicolumn{1}{l|}{$7.2 \cdot 10^{6}$} & \multicolumn{1}{l|}{$1.408 \cdot 10^{-1}$} & $3.274$ & \multicolumn{1}{l|}{$6 \cdot 10^{6}$} & \multicolumn{1}{l|}{$2.139 \cdot 10^{-1}$} & $5.850$ \\
11 & \multicolumn{1}{l|}{$7.2 \cdot 10^{6}$} & \multicolumn{1}{l|}{$1.227 \cdot 10^{-1}$} & $3.616$ & \multicolumn{1}{l|}{$6 \cdot 10^{6}$} & \multicolumn{1}{l|}{$1.583 \cdot 10^{-1}$} & $7.168$                          
\end{tabular}
}
\end{table}

For each test for dimension from $3$ to $8$, we generated $2500$ examples of distance matrices $D$ and for each $D$ we generated $9600$ triples of random vectors (test 1 and 2) or $8000$ triples of orthonormal vectors (test 3 and 4). By the procedure described above and recorded the minimal difference between RHS and LHS of \eqref{triangins}.
Moreover, for each dimension from $8$ to $11$, we generated $7500$ examples of distance matrices $D$ and for each $D$ we generated $9600$ triples of random vectors (test 1 and 2) or $8000$ triples of orthonormal vectors (test 3 and 4) and recorded the minimal difference between RHS and LHS of \eqref{triangins}.
The obtained results are presented in the Table \ref{tab:num1}. In generated trials we did not find any counterexample.

The simulations were programmed in the Python language, using the Numpy library for algebraic calculations with standard float arithmetic. The largest errors in our calculations originated from the generation of random orthonormal vectors and were of the order of $10^{-15}$, thus we take $10^{-14}$ as a numerical accuracy of our algorithm.


\begin{thebibliography}{MMM}
\bibitem{BZ17} I. Bengtsson and K. {\.Z}yczkowski, 
{\sl  Geometry of Quantum States}. 2 ed., Cambridge University Press, 
Cambridge 2017.
\bibitem{QML} J. Biamonte, P. Wittek, N. Pancotti, P. Rebentrost, N. Wiebe, N., and S. Lloyd, Quantum machine learning, {\sl Nature} \textbf{549} (2017), 195.

\bibitem{BEZ22} R. Bistro\'n, M. Eckstein, K. \.Zyczkowski, Monotonicity of the quantum 2-Wasserstein distance, 
{\sl J. Phys.} {\bf A 56},  (2023) 095301.


\bibitem{CGP20} E. Caglioti, F. Golse, and T. Paul,
 Quantum optimal transport is cheaper,
  {\sl J. Stat. Phys.}, {\bf 181} (2020), 149.

\bibitem{CEFZ21} S. Cole, M. Eckstein, S. Friedland, K. {\.Z}yczkowski,  On Quantum Optimal Transport, \textit{Math. Phys., Analysis and Geometry} \textbf{26} (2023), 14.

\bibitem{QSense} Ch.L. Degen, R. Friedemann and P. Cappellaro, Quantum sensing, {\sl Reviews of Modern Physics} \textbf{89} (2017), 035002.

\bibitem{Duv22} R. Duvenhage, Quadratic Wasserstein metrics for von Neumann algebras via transport plans, 
{\sl J. Operator Theory} {\bf  88} (2022), 289-308.

\bibitem{Frib}  S. Friedland, \emph{Matrices: Algebra, Analysis and Applications}, World Scientific, 596 pp., 2016, Singapore, http://www2.math.uic.edu/$\sim$friedlan/bookm.pdf

\bibitem{FECZ21}  S. Friedland, M. Eckstein, S. Cole and K. \.Zyczkowski, Quantum Monge-Kantorovich problem and transport distance between density matrices,  
\emph{Phys.\ Rev.\ Lett.} {\bf 129} (2022), 110402.

\bibitem{QMetro} V. Giovannetti, S. Lloyd and L.Maccone, Advances in quantum metrology, {\sl Nature Photonics} \textbf{5}  (2011), 222--229.

\bibitem{GP18} F. Golse and T. Paul,
 Wave packets and the quadratic Monge-Kantorovich distance in quantum mechanics,
 {\sl  Comptes Rendus Math.} {\bf 356} (2018),  
 177-197.

\bibitem{HJ13} R.A.Horn and C.R. Johnson, 
\emph{Matrix analysis},
Second edition, Cambridge University Press, Cambridge, 2013. 

\bibitem{Kan48} L.V. Kantorovich, On a problem of Monge,
 {\sl Uspekhi Mat. Nauk.} {\bf 3} (1948), 225.

\bibitem{Keyl} M. Keyl, Fundamentals of quantum information, \emph{Phys. Rep.} \textbf{369} (2002), 431--548.

\bibitem{KdPMLL22} B.T. Kiani, G. De Palma, M. Marvian, Z.-W. Liu, and S. Lloyd,
 Quantum Earth Mover's Distance: A New Approach to Learning Quantum Data, 
 Quantum Science and Technology {\bf 7},  (2022) 045002.


\bibitem{Mae13}  H. Maehara,  Euclidean embeddings of finite metric spaces, \emph{Discrete Math. } {\bf313} (2013), no. 23, 2848-2856.


\bibitem{Mi23} T. Miller,
On the diagonal product of special unitary matrices,
(2023) {\sl to be published}


\bibitem{NielsenChuang} M.A. Nielsen and I.L. Chuang, \emph{Quantum Computation and Quantum Information}, Cambridge University Press (2010).

\bibitem{PMTL21} G.  De Palma, M. Marvian, D. Trevisan and S. Lloyd, 
The Quantum Wasserstein Distance of Order $1$, 
\textit{IEEE Transactions on Information Theory} {\bf 67} (2021), 6627 - 6643.

\bibitem{PT21} G.  De Palma and D. Trevisan, 
 Quantum Optimal Transport with Quantum Channels,  
Ann. Henri Poincar\'e {\bf 22} (2021), 3199-3234.

\bibitem{Roc70}  R.T. Rockafellar, {\it Convex Analysis}, Princeton Mathematical Series, No. 28 Princeton University Press, Princeton, N.J. 1970 xviii+451 pp.

\bibitem{Sch35} I.J.  Schoenberg, Remarks to Maurice Frechet’s article "Sur La Definition Axiomatique D'Une Classe D'Espace Distances Vectoriellement Applicable Sur L'Espace De Hilbert", Ann. of Math. 36 (1935), 48-55.

\bibitem{Sch37}  I.J. Schoenberg,  On certain metric spaces arising from Euclidean spaces by a change of metric and their imbedding in Hilbert space, \emph{Ann. of Math. }(2) \textbf{38} (1937), no. 4, 787--793. 

\bibitem{Sch38} I.J.  Schoenberg, Metric spaces and positive definite functions. \textit{Trans. Amer. Math. Soc.} \textbf{44} (1938), 522--536.

\bibitem{ZYYY22}  L. Zhou, N.Yu, S. Ying, M. Ying,
Quantum Earth mover's distance,
No-go quantum Kantorovich-Rubinstein theorem, and quantum marginal,
 \textit{J. Math. Phys.} \textbf{63} (2022), 102201.

\bibitem{ZS01} K. {\.Z}yczkowski and W. S{\l}omczy{\'n}ski,
The Monge distance on the sphere and geometry of quantum states,
{\sl  J. Phys. A} {\bf 34} (2001), 6689.


\end{thebibliography}
\end{document}